\documentclass[12pt]{alt2022} %

\SetKwComment{Comment}{// }{}

\allowdisplaybreaks

\def\FastClock{~\emph{FastClock}}
\def\Nbd{\mathcal{N}}

\def\union{\bigcup}
\def\intersect{\bigcap}

\def\Binomial{\mathrm{Binomial}}
\def\Dis{\mathrm{Dis}}
\def\E{\mathbb{E}}
\def\CF{\mathcal{CF}}
\def\CCF{\mathcal{CCF}}
\def\symmdiff{\triangle}
\def\SEst{\tilde{S}}
\def\F{\mathbb{F}}
\def\Frontier{\F}

\def\N{\mathbb{N}}

\def\lunion{\cup}
\def\lintersect{\cap}
\def\C{\mathcal{C}}

\def\T{\mathcal{T}}
\def\R{\mathbb{R}}
\def\Z{\mathbb{Z}}
\def\FrontierDisc{\Delta\Frontier} 
\def\RDisc{\Delta R}

\newcommand{\TL}[1]{[[#1]]}

\newtheorem{assumption}{Assumption}

\title[ Network Cascade Temporal Scale Estimation ]{ Temporal Scale Estimation for Oversampled Network Cascades: Theory, Algorithms, and Experiments }
\usepackage{times}

\altauthor{%
 \Name{Abram Magner} \Email{amagner@albany.edu}\\
 \addr University at Albany, SUNY
 \AND
 \Name{Carolyn Kaminski} \Email{ckaminski@albany.edu}\\
 \addr University at Albany, SUNY%
 \AND
 \Name{Petko Bogdanov} \Email{pbogdanov@albany.edu} \\
 \addr University at Albany, SUNY
}

\begin{document}

\maketitle

\begin{abstract}
    Spreading processes on graphs arise in a host of application domains, 
    from the study of online social networks to viral marketing to epidemiology.
    Various discrete-time probabilistic models for spreading processes
    have been proposed. These are used for downstream statistical estimation
    and prediction problems, often involving messages or other information
    that is transmitted along with infections caused by the process.  It
    is thus important to design models of cascade observation that take
    into account phenomena that lead to uncertainty about the process
    state at any given time.
    We highlight one such phenomenon -- \emph{temporal
    distortion} -- caused by a misalignment between the rate at which observations
    of a cascade process are made and the rate at which the process itself operates, and argue that failure to correct for it results in degradation
    of performance on downstream statistical tasks.  
    To address these issues, we formulate the \emph{clock estimation} problem
    in terms of a natural distortion measure.
    We give a clock estimation algorithm, which we call FastClock, that runs in
    linear time in the size of its input and is provably statistically
    accurate for a broad range of model parameters when cascades are generated
    from the independent cascade process with known parameters and when
    the underlying graph is Erd\H{o}s-R\'enyi. 
    We further give empirical results on the performance of our algorithm
    in comparison to the state of the art estimator, a likelihood proxy maximization-based estimator implemented via dynamic programming.  We find that, in a broad parameter regime, our algorithm substantially outperforms the dynamic programming algorithm in terms of both running time and accuracy.
\end{abstract}


\section{Introduction}

There are a variety of well-established and simple probabilistic generative models for graphs and 
and infectious processes that run over these graphs. In this work we specifically focus on models 
for spreading processes on networks such as the diffusion of 
innovation~\cite{montanari2010spread}, information~\cite{bakshy2012role} and 
misinformation~\cite{shin2018diffusion} in social networks. Accurate estimation of model 
parameters of such processes based on observational data is essential for a variety of important 
applications: from product marketing and social network recommendations to studying financial 
markets and detecting 
insurgent networks and limiting misinformation. At the same time, accurate modeling critically 
depends on our ability to account for major sources of uncertainty induced by the manner in which 
observational data about such evolving processes is acquired. 

\paragraph{Discrete-time diffusion process models.}
Several well-studied information diffusion models assume a discrete timeline in which at every time step nodes participate in the diffusion process (i.e., get ``infected'') based on influence from network neighbors who got infected in past time steps. For example, according to the \emph{independent cascade model}~\cite{KempeInfluenceMaximization} infected nodes have one chance to infect their neighbors, while in the \emph{linear threshold model} nodes get infected when a critical fraction of their neighbors have been infected in any prior time steps. 

It is important for our subsequent discussion to note what ``discrete-time'' means in the context 
of a process running in the real world, about which we would like to draw statistical inferences 
based on observations at potentially arbitrary physical time points.  We think of a diffusion
process as running in continuous time, so that, in principle, a vertex infection may occur at any 
$t\in \R$.  A discrete-time process model posits the existence of a sequence of (possibly random) \emph{time steps} $0 \leq \tau_0 < \tau_1 < ..., \tau_j \in [0, \infty)$, and specifies the probability distribution of the process state at each time $\tau_{j+1}$
conditioned on the state of the process at time $\tau_{j}$.  Each such conditional distribution is invariant to the actual values of the $\tau_j$.  In this sense, we can think of a discrete-time process model as a partial specification of a continuous-time process model whose state evolves according to a discrete-valued variable.

\paragraph{The need to account for temporal distortion.}
One major source of uncertainty
that is overwhelmingly overlooked in current literature is a misalignment
of time points at which we observe a discrete-time process trajectory with 
the time points at which the state variables governing the process evolve.
Here, an \emph{observation} of a process at a particular time consists
of the current state (infected or not) of every node.  The aforementioned
misalignment may be, for example, due to drawing observations at a higher
rate than that at which the cascade process itself operates.
This leads to what we call \emph{temporal distortion} in process 
observations.  Correcting for this distortion is the \textbf{main focus} of this paper.  We illustrate this phenomenon with a concrete example, 
Example~\ref{example:temporal-distortion}, that shows the deleterious 
effect of uncorrected temporal distortion on a downstream statistical 
estimation task.

\begin{example}[Temporal distortion affects downstream statistical inference]

    Consider a cascade generated by the \emph{independent cascade model}~\cite{KempeInfluenceMaximization}
    on a graph $G$ with $n$ vertices, with edge transmission parameter $p_n=1$ and probability of infection from an external source $p_e = 0$.  Assume that $G$ is a complete binary tree and that the infection starts at the root node.
    We recall that this model runs in discrete time, with physical timesteps
    $t_0 = 0 < t_1 < ... < t_N$, with $t_j \in \R$ for all $j$, producing
    infected vertex sets $S_j\subseteq [n] = \{1, ..., n\}$, for each 
    $j \in \{0, ..., N\}$.  That is, $S_j$ is the set of vertices infected
    in the physical time interval $(t_{j-1}, t_j]$.  In each time interval $(t_{j-1}, t_j]$, the set of \emph{active} vertices (those vertices that can transmit infections across edges) is $S_{j-1}$.  Let us suppose that the infection times of vertices
    infected in a given time interval are uniformly distributed in that interval.
    For this example, we choose 
    physical observation times $t'_0=0, t'_1, ..., t'_{2N+1}$ with
    $t'_{2j} = t_j$ and $t'_{2j+1} = \frac{t_{j} + t_{j+1}}{2}$ for each
    $j$.
    Thus, our view of the cascade consists of a sequence of infected sets $\hat{S}_j$, $j \in \{0, 1, ..., 2N+1\}$.  
    
    Consider the problem of \emph{doubling time} prediction: given cascade
    observations up to/including a time $t \in \R$ in which $m$ vertices are infected, the task is to predict
    an interval $[a, b]$ such that, with probability at least $1-\delta$,
    for some fixed parameter $\delta > 0$, the physical time of the $2m$-th infection event lies in $[a, b]$.
    
    If temporal distortion is \emph{not} accounted for, so that we incorrectly 
    assume that the process timesteps occur at times $t'_j$, we have an 
    inaccurate knowledge of the set of active vertices at any given time.
    This has the following effect on doubling time prediction at
    times $t = t'_{2j+1}$: at these times, approximately $\sum_{k=0}^{2j} 2^k + 2^{2j+1} = 2^{2j+1} + \frac{2^{2j+1} - 2^{0}}{2j+1-0} = \frac{2^{2j+1} - 1}{2j+1} + 2^{2j+1} = 2^{2j+1} \cdot \left(1 + O(1/j) \right)$ vertices are infected, and we believe that approximately $2^{2j+1}$ vertices are active (when, in fact, only $O(2^{2j}/j) = o(2^{2j+1})$ vertices are active).  We would thus predict that the doubling time is exactly $t'_{2(j+1)}$, despite
    the fact that the number of vertices infected at this time is
    exactly $\sum_{k=0}^{2(j+1)} 2^k = \frac{2^{2(j+1)+1}-2^{0}}{2j+2} = \frac{2^{2j+3} - 1}{2j+2} = \frac{2^{2j+2}-1}{j+1} \ll 2^{2j+2}$.  Thus,
    failure to account for temporal distortion in this setting leads to 
    substantial and, in this setting, avoidable inaccuracy.
    
    More realistic empirical experiments in~\cite{ditursi2017network,ditursi2019optimal} confirm
    that accounting for temporal distortion can, in practical settings,
    improve performance on doubling time prediction and several other
    downstream statistical tasks.
    \label{example:temporal-distortion}
\end{example}

More generally, temporal distortion degrades statistical performance
on problems where model parameters are dependent on knowledge of the
infectious sets of vertices (the so-called \emph{active} vertices mentioned in the example) at given times.  Correction for temporal
distortion, which is the main focus of the present paper, is thus
an important problem.

\paragraph{Prior work.}
The general topic of analysis of cascades has received a large amount
of attention, both from theoretical and empirical perspectives.
There are many cascade models, with features depending on application
domains.  E.g., the independent cascade (IC) and linear threshold (LT)
models were popularized in~\cite{KempeInfluenceMaximization} for the
application of \emph{influence maximization}.  This problem continues to be studied~\cite{AbbeInfluence1,AbbeInfluence2}.
Variations on the influence maximization problem that have
time-critical components and, thus, may be sensitive to temporal
perturbation in the sense that we study here, have also been
studied~\cite{TimeCritical,ali2019fairness}.
These models are
also used outside the context of influence maximization, e.g., in the modeling of the spread of memes on social networks~\cite{AdamicMemes}. 

Statistical prediction tasks involving cascades have also been
asked.  For instance, the cascade doubling time prediction task was considered in~\cite{CanCascadesBePredicted}.  Other works
propose models in which a piece of information, such as a message, an opinion, or a viral genome, is transmitted
along with the infection of a node~\cite{PoorViralEvolution,LearningForecastingOpinions,GenerationIntervals}.
For such statistical problems, statistical inferences about
the transmitted information can be disrupted by inaccurate estimation of the set of infectious vertices at a given time, further motivating the consideration of methods for correcting
for temporal distortion.

In~\cite{ditursi2017network} (see also followup work 
in~\cite{ditursi2019optimal}), the authors formulated a version of the 
problem of clock recovery (equivalent to temporal distortion correction studied here)
from adversarially temporally perturbed cascade data as a problem of maximization of
a function of the clock that serves as a proxy (in particular, an upper bound) for the log likehood of the observed
cascades.  They proposed a 
solution to this problem via a dynamic programming algorithm.  While 
the dynamic
programming algorithm is an exact solution to their formulation of the problem,  
it has a running time of $\Theta(n^4)$, where
$n$ is the total number of vertices in the graph on which the observed
cascade runs, which is
prohibitively expensive for graphs of moderate to large size.  Furthermore, 
their formulation of the problem makes no comparison with the ground truth clock,
and thus there are no theoretical guarantees or empirical evaluations of the accuracy of their estimator (which we call the \emph{maximum likelihood proxy (MLP) estimator}) as an approximation to the
ground truth clock.  In contrast, the present work gives a rigorous formulation of the
problem as one of statistical estimation of the ground truth clock from observed cascades.
We compare our proposed algorithm and estimator with the MLP estimator in this framework
in terms of both accuracy and running time.

\paragraph{Our contributions}

In the present work, we propose an approximation formulation of the clock recovery
problem, allowing us to quantify the proximity of estimated clocks to the ground
truth in a principled manner.  Our formulation is general, covering arbitrary varying observation rates. However, our algorithms, theorems,
and experiments are specific to the \emph{oversampling} case, wherein observations
are made at a higher rate than that at which the spreading process operates.  We
leave estimators for the more complicated general case to future work.

We propose a novel estimation algorithm, which we call
FastClock, that runs in time linear in the size of the cascade.  We rigorously prove
that,
under natural conditions on the input graph and cascade parameters, the FastClock
estimator produces a clock whose distance to the ground truth is vanishingly small
as the size of the graph tends to infinity.  Our guarantees on FastClock
hold for a broad range of the parameter of the Erd\H{o}s-R\'enyi graph 
model. 

We bolster our theoretical results via experiments on synthetic graphs
and cascades.  We find that the FastClock estimator empirically outperforms
the dynamic programming-based estimator from~\cite{ditursi2017network} in these experimental conditions in terms of accuracy and 
running time.

\paragraph{Organization of the paper}
In Section~\ref{sec:problem-formulation}, we give a precise formulation of the problem and introduce notation.  In Section~\ref{sec:main-results}, we state the FastClock algorithm and the main
theoretical results.  We give proof sketches (and, where noted, full proofs) in
Section~\ref{sec:proof-sketches}.  Section~\ref{sec:empirical} gives empirical results comparing
FastClock and the DP algorithm implementing the MLP estimator.   We conclude in 
Section~\ref{sec:conclusions}.  Full proofs
of all results are given in the appendix.

\section{Problem formulation and notation}
\label{sec:problem-formulation}

Our goal in this section is to formulate the problems of \emph{clock estimation} and \emph{spreading process history reconstruction}
from a temporally perturbed cascade observation.  As mentioned in the introduction,
our formulation is quite general and covers temporal distortion arising from
arbitrarily varying observation rates.  Since this general case is algorithmically
and statistically more complicated (in particular, while our proposed algorithm succeeds at clock estimation, the more relevant problem of history reconstruction is more difficult), we then focus on the oversampling case.  In this case,
our general problem formulation can be replaced by a simpler one,
and the two problems of clock estimation and spreading process history reconstruction become equivalent.

\subsection{General formulation}

We fix a graph $G$
on the vertex set $[n] = \{1, ..., n\}$, and we define the \emph{timeline of
length $N$}, for any number $N \in \N$,
to be the set $\TL{N} = \{0, 1, ..., N\}$.
The first ingredient of our framework is a \emph{cascade model}.
\begin{definition}[Cascade model]
    A (discrete-time) cascade model $\C(N)$ is a probability distribution 
    on sequences
    $(S_0, S_1, ..., S_N)$ of disjoint subsets of vertices of $G$.
    We think of $S_t$, $t \in \{0, ..., N\}$ as the set of vertices
    infected in timestep $t$.  We call any such sequence
    an \emph{infection sequence}, and we write $|S| = N+1$.
\end{definition}
As mentioned in the introduction, we think of a discrete-time cascade as running in
continuous, physical time, so that the $j$th timestep begins at some physical time $t_{j-1}$
and ends at physical time $t_j$, and every vertex $v \in S_j$ is infected at some
physical time in the interval $(t_{j-1}, t_j]$.  Note, however, that physical times
are not formally part of the logical framework, and our models have no explicit dependence on them.  
We introduce them only to aid intuition.

Next, we define our observation model, which formalizes our notion of
temporal perturbations.  To do this, we need the notion of a \emph{clock}.
Intuitively, a clock encodes the number of \emph{observations} of the cascade
made during each cascade timestep.  For us, an observation of
a cascade at some physical time $t$ consists of the set of nodes that have been infected at or before time $t$.  We will talk about the $k$th observation to occur, $k\geq 0$, as having
\emph{index} $k$.

\begin{definition}[Clock]
    A \emph{clock} $C$ on a timeline $\TL{N}$ is a map $C:\TL{N}\to\Z^{\geq 0}$.
    Equivalently, it is a tuple of $N+1$ non-negative integers $(C(0), ..., C(N))$,
    where $C(j)$ intuitively gives the number of observations made in the physical time interval
    $(t_{j-1}, t_j]$.  The \emph{size} $|C|$ of $C$ is given by 
    \begin{align}
        |C| = \sum_{j=0}^N C(j).
    \end{align}
    \label{def:general-clock}
\end{definition}
It will be convenient to introduce more notation regarding clocks:
\begin{itemize}
    \item
        For a clock $C$, let the $j$th partial sum of $C$ be given by $\sum_{k=0}^j C(k)$, and 
        denote it by $C(0:j)$.  This is the number of observations made prior to the $j$th
        cascade timestep.
    \item
        Let $M_{C}:\TL{N}\to 2^{\TL{N'}}$ be given as follows: 
        $M_{C}(j) = \{ C(0:j-1)+1, ..., C(0:j)  \}$.
        In other words, $M_{C}(j)$ is the set of observation indices that occur during the time 
        interval $(j-1, j]$, according to $C$.
\end{itemize}

The following definition captures the notion of an infection sequence $S'$ that could arise
from observing a ground truth infection sequence $S$ according to a schedule dictated
by a clock $C$.
\begin{definition}[Clock-consistent observation of an infection sequence]
    Fix two infection sequences $S = (S_0, ..., S_N)$ and $S' = (S'_0, ..., S'_{N'})$
    and a clock $C$ on $\TL{N}$ with size $|C| = N'$.  We say that $S'$ is an observation
    of $S$ that is consistent with $C$ if,
    for each ground truth timestep $t \in \TL{N}$,
    \begin{align}
        S_t
        = \union_{t' \in M_C(t)} S'_{t'}.
    \end{align}
    In other words, $S'_{j}$ can be interpreted as encoding $j$th observation of the infection 
    sequence given by $S$, according to the schedule dictated by $C$.
\end{definition}
As an easy consequence of this definition, if $S'$ is an observation of $S$ consistent with
any clock $C$, then $C$ is the unique clock for which this is true.

\begin{example}[Infection sequences, clocks, clock-consistent observations]
    Consider a graph $G$ on the vertex set $[n] = [10]$.  One possible infection
    sequence on $G$ is 
    \begin{align}
        S
        = (S_0, S_1, S_2) 
        = (\{2, 8, 10 \}, \{ 1, 3, 4, 7, 9 \}, \{ 6 \} ).
    \end{align}
    This encodes a sequence of infections occurring in three timesteps --
    i.e., on the timeline $\TL{2}$.
    In particular, we may think of $S_1$ as encoding that vertices
    $1, 3, 4, 7, 9$ are all infected during timestep $1$, but the
    order in which they are infected is not encoded.
    
    One possible example clock on the timeline $\TL{2}$ is
    $ %
        C
        = (C_0, C_1, C_2)
        = (0, 4, 1).
    $ %
    This encodes that $0$ observations are made in timestep $0$,
    $4$ are made in timestep $1$, and $1$ is made in timestep $2$.
    
    An example infection sequence $\hat{S}$ that is an observation of $S$
    consistent with $C$ is as follows:
    \begin{align}
        \hat{S} = (\hat{S}_0, \hat{S}_1, \hat{S}_2, \hat{S}_3, \hat{S}_4)
        = (\{ 2, 7, 8, 10 \}, \{ 1, 9 \}, \{ 3, 4 \}, \{ \}, \{ 6 \})
    \end{align}
    Note that $\hat{S}$ is necessarily an infection sequence on the
    timeline $\TL{0 + 4 + 1 - 1} = \TL{4}$.
\end{example}

We finally come to the definition of a temporal distortion model.
\begin{definition}[Temporal distortion model]
    A temporal distortion model is a conditional probability distribution $P(\cdot ~|~ S)$ on infection sequences, parameterized by infection sequences $S$ (which we think of as being the ground truth infection sequences), such that
    $P(S' ~|~ S) > 0$ only if $S'$ is an observation of $S$ consistent with some clock.
\end{definition}
In other words, a temporal distortion model is a probabilistic generative model for
observations of an infection sequence.

\subsection{Specialization to the oversampling regime}

\textbf{In this work, we will focus without much further comment on \emph{oversampling} temporal distortion models,}
which are models resulting in observations according to clocks with $C(j) > 0$ for all $j$.
Intuitively, this covers the case where observations are made at a higher rate than that
at which the process itself operates.  In the oversampling regime, we can simplify the
definition of a clock:
\begin{definition}[Clock (oversampling case)]
    A(n oversampling) clock $C$ on the timeline $\TL{N}$ with size $N'+1$ is a partition of $\TL{N'}$
    into $N+1$ subintervals.  We call the $j$th such subinterval,
    for $j=0$ to $N$, the \emph{$j$th observation interval}.
    \label{def:clock-oversampling}
\end{definition}
In the above definition, we think of $\TL{N}$ as the ground truth timeline and
$\TL{N'}$ as the observation timeline.  An oversampling clock partitions the observation timeline
into subintervals, each corresponding to a single ground truth timestep.  

\begin{example}[Oversampling clock]
    Consider the timeline $\TL{5}$ (here, $N' = 5$).  An example of an oversampling clock is
    \begin{align}
        C = \{ [0, 2], [3, 3], [4, 5] \}.
    \end{align}
    This is equivalent to the following clock on $\TL{2}$, with $N=2$, in the
    sense of Definition~\ref{def:general-clock}:
    \begin{align}
        C = (3, 1, 2).
    \end{align}
\end{example}

An infection sequence $S$ naturally induces a partial order on the set of vertices: namely,
for two vertices $a, b$, $a < b$ if and only if $a \in S_i, b \in S_j$ for some
$i < j$.  Similarly, a clock on an infection sequence, in the sense of Definition~\ref{def:clock-oversampling}, induces a partial order.

We will consider two clocks $C_0, C_1$ to be equivalent with respect to a given observed infection sequence $S$ if they induce the same partial order.  The reason
for this is that two equivalent clocks separate vertices in the same way into a
sequence of time steps.  We will sometimes abuse terminology and use ``clock''
to mean ``clock equivalence class''.

We next define a distortion function on clock equivalence classes.  This will allow
us to measure how far a given estimated clock is from the ground truth.
Note that given an observed infection sequence $S$, a clock cannot reverse the order
of any pair of events, so that the standard Kendall $\tau$ distance
between partial orders is not appropriate here.  
\begin{definition}[Distortion function on clock pairs]
\label{def:dist}
    Consider two clocks $C_0, C_1$ with respect to an observed infection 
    sequence $S$.
    We define $\Dis_{C_0,C_1}(i,j)$ to be the indicator that the clocks 
    $C_0$ and $C_1$ order vertices $i$ and $j$ differently (i.e., that
    the partial order on vertices induced by $C_b$ orders $i$ and $j$
    and the partial order induced by $C_{1-b}$ does not, for $b$ equal
    to either $0$ or $1$).
    If the clocks in question are clear from context, we may drop the 
    subscript.
    
    We define the following distortion measure on clock pairs:
    \begin{align}
        d_S(C_0, C_1)
        = \frac{1}{{n\choose 2}} \sum_{i < j} \Dis_{C_0, C_1}(i, j).
    \end{align}   
\end{definition}

We finally come to the general problem that we would like to solve:
\begin{definition}[Clock estimation/Spreading process history reconstruction]
    Fix a graph $G$, a cascade model $\C(G, T)$, and an oversampling temporal distortion model $\T$.  An infection sequence $S \sim \C(G, T)$ is generated on $G$.  Finally, an observed infection sequence
    $\hat{S}$ with $|\hat{S}| = N+1$ is generated according to $\T$, with implicit clock $C$.  Our goal is to produce an estimator $\hat{C} = \hat{C}(\hat{S})$ of $C$ so
    as to minimize $\E[d_{\hat{S}}(C, \hat{C})]$.  This is called the \emph{clock estimation problem}.
    
    We
    call the problem of estimating $S$ the \emph{spreading process history reconstruction problem}.  An estimated oversampling clock $\hat{C}$ induces an estimate $\SEst$ of $S$,
    so that clock estimation and spreading process history reconstruction are equivalent
    in the oversampling case.
\end{definition}
We note that the above definition implicitly assumes knowledge of the parameters of
the cascade model.  Estimation of these parameters has been studied in the literature.
Furthermore, we note that knowledge of the initial conditions of the cascade is necessary
in order to achieve an expected estimation error that tends to $0$ in general.  We thus assume
that the number of initially infected vertices is given to us. Under mild additional
assumptions on the model (e.g., that $S_0$ consists of $\Theta(1)$ vertices chosen uniformly at random, and that the graph is sparse, so that $S_0$ is an independent set with high probability), the initial set $S_0$ can be inferred with high probability.

\paragraph{Specific cascade models.}
Having laid out the general framework for temporal distortion models, we specify a few 
example cascade models for our problem.  Our approach generalizes beyond these two,
as we will explain after the statement of our algorithm.

We first define the \emph{independent cascade (IC)} process.  We fix a graph $G$,
an initial infection set $S_0$ of vertices in $G$ (given by elements of
$[n] = \{1, ..., n\}$), and probability parameters 
$p_n$ and $p_e$.  Here, $p_n$ denotes the probability of transmission of an
infection across an edge, and $p_e$ denotes the probability of infection from
an external source.

Step $j+1$ of the IC process proceeds as follows: for each node $v$ in $S_j$
and each uninfected neighbor $w$, $v$ attempts to infect $w$, succeeding with
probability $p_n$, independent of anything else.  Next, each uninfected node
is independently infected with probability $p_e$.  The set of nodes infected
in step $j+1$ is denoted by $S_{j+1}$.
The process terminates either after a specified number of steps, after all
nodes are infected, or when $S_{j+1} = \emptyset$ and $p_e = 0$.

The \emph{linear threshold (LT)} process works as follows: for every node $v$ in $G$, a threshold
$\theta_v$ is drawn independently from some known distribution on $[0, 1]$.  Some initial subset $S_0$
of vertices is infected, and, in each subsequent timestep, vertex $v$ is infected if
either it has already been infected or the fraction of its neighbors that are infected
exceeds $\theta_v$.

\section{Main results: Algorithm, approximation and running time guarantees}
\label{sec:main-results}

In this section, we present our proposed algorithm (Algorithm~\ref{alg:FastClock}) for clock 
estimation, which we call\FastClock. It takes as input a graph $G$, an observed infection 
sequence $\hat{S} = (\hat{S}_0, ..., \hat{S}_N)$, and the parameters $\theta$ of the 
cascade model, including the initial infection set $S_0$ (see our discussion of this 
assumption in the previous section).  The output of the algorithm is an estimated clock
$\hat{C}$, which takes the form of a sequence of interval endpoints $\hat{t}_0, \hat{t}_1, ..., \hat{t}_{\hat{N}} \in \TL{N}$, for some $\hat{N}$ and is an estimate of the ground truth clock $C$ specified by $t_0, ..., t_N$.

Our algorithm proceeds by iteratively computing the estimate $\hat{t}_{j}$.  In the
$(j+1)$-st iteration, to compute $\hat{t}_{j+1}$, 
it chooses the size of the next interval of the clock
so as to match as closely as possible the expected
number of newly infected nodes in the next timestep.  We will 
prove that the resulting clock estimate is very close, in terms 
of $d_{\hat{S}}(\cdot, \cdot)$, to the ground truth clock, using 
concentration inequalities.

In particular, the correctness of FastClock is based on the following intuition: if we manage to correctly
estimate $t_0, ..., t_j$, then we can estimate the conditional expected number of
vertices infected in the $(j+1)$-st timestep of the process (i.e., $|S_{j+1}|$).  We can
show a conditional concentration result for $|S_{j+1}|$ around its expectation.
Thus, we output as our next clock interval endpoint $\hat{t}_{j+1}$ the smallest integer
$t \geq \hat{t}_{j}$ for which the number of vertices in $\union_{k=t_{j}+1}^{t} \hat{S}_{t}$
does not exceed its conditional expectation, corrected by a small quantity.  This 
quantity is determined by the concentration properties of the random variable $|S_{j+1}|$
conditioned on the state of the process given by $\hat{S}_0, ..., \hat{S}_{j}$.  We choose
the threshold to be such that, under this conditioning, the number of vertices infected
in the next process timestep is slightly less than it with probability tending exponentially
to $1$.
Our approximation analysis illustrate that the
approximation quality depends on the graph structure and the model parameters.

The significance of the approximation and running time results (Theorems~\ref{thm:fastclock-approx-main} and~\ref{thm:fastclock-running-time} below) is that oversampling temporal distortion under natural conditions can be quickly corrected for
with provably high accuracy using relatively simple expected value calculations.  While our approximation theorem is formally stated for the IC model, the conclusions hold as long as the number of infected nodes
in the next cascade timestep, conditioned on the current state of
the process, is well-concentrated and as long as the
\emph{expected} number of such nodes is immune to small
errors in the estimation of the process state.  These
are both functions of the cascade model and of the 
structure of the graph $G$ on which the cascade runs: our results
hold when the graph is an expander with appropriate parameters
(which is implied by our Erd\H{o}s-R\'enyi stipulation in the approximation theorem).

As long as the
expected number of nodes infected in the next timestep can be calculated efficiently,
the FastClock algorithm can be adapted to a wide variety of cascade models.

\subsection{The FastClock algorithm}
Before we define our algorithm we introduce some necessary notation. For an infection sequence $\SEst$ and a timestep index $t \in |\SEst|$, define $\sigma_t(\SEst)$ to be the $\sigma$-field
generated by the event that the first $t$ infection sets of
the cascade process are given by $\SEst_0, \SEst_1, ..., \SEst_t$.  That is, the event in question is that
$S_0 = \SEst_0, ..., S_t = \SEst_t$.
We also define $\mu_{t}(\SEst)$ to be
$ %
\mu_t(\SEst)
= \E[ |S_{t+1}| ~|~ \sigma_t(\SEst)].
$ %
The algorithm is given in Algorithm~\ref{alg:FastClock}.

\begin{algorithm2e}[th]
\label{alg:FastClock}
\caption{FastClock}
\KwData{Graph $G$, cascade model parameters $\theta$, observed infection sequence $\hat{S} = (\hat{S}_0, ..., \hat{S}_N)$}
\KwResult{An estimated clock $\hat{C}$.} 
\Comment{An initially empty list for the estimated clock.} 
Set $\hat{C} = ()$;

\Comment{$t$: index of the next estimated clock interval, i.e., $t$ is an index in $S$, the un-distorted infection sequence.  \\
$t_{obs}$: the index in $\hat{S}$ of the beginning of the next estimated clock interval %
}
Set $t=1, t_{obs}=\min\{j \leq N ~:~ |\lunion_{k=0}^j \hat{S}_k| = S_0\}$\;

\Comment{$\SEst$: the estimated infection sequence approximating the ground truth sequence $S$.}
Set $\SEst_0 = \lunion_{k=0}^{t_{obs}} \hat{S}_k$\;

Append $t_{obs}$ to $\hat{C}$\;

\While{$t_{obs} \neq N$}{
\Comment{Compute the expected number $\mu_t$ of infected nodes %
in a single step of the cascade process, starting
from the state of the process estimated so far.}
Set $\mu_t = \E[ | S_{t+1} | ~|~ \sigma_{t}(\SEst_0, \SEst_1, ..., \SEst_{t}) ]$\;

Set    
\begin{align}
t'_{obs}
= t_{obs} + \max\left\{\Delta ~|~  \sum_{i=t_{obs}+1}^{t_{obs+\Delta}} |\hat{S}_i|  \leq \mu_t \cdot (1 + \mu_t^{-1/3}) \right\}\;
\end{align}

Append $t'_{obs}$ to $\hat{C}$\;

Set $\SEst_t = \lunion_{i=t_{obs}+1}^{t'_{obs}} \hat{S}_i$\;

Set $t = t+1$\;

Set $t_{obs} = t'_{obs}$\;
} 

\Return{$\hat{C}$}\;
\end{algorithm2e}

After an initialization, the main loop in \FastClock~  (Steps 5-11) iteratively determines the first infection event in the next step of the process, by estimating the expected number of nodes $\mu_t$ to be infected next (Step 6). The key step in this process is the computation of $\mu_t$ which we discuss next.   

\paragraph{Computing $\mu_t(\SEst)$ in the IC model.}
Let us be more precise in specifying how to compute $\mu_t(\SEst) = \E[|S_{t+1}| ~|~ \sigma_{t}(\SEst)]$ in the independent cascade model.  A node can be infected in
one of two ways: through external factors (governed by $p_e$) or via transmission from
a vertex in $\SEst_t$ through an edge.  In the latter case, the node must lie in the \emph{frontier set}
$\Frontier_t(\SEst)$, defined as follows: $\Frontier_t(\SEst)$ is the
set 
$ %
    \Frontier_t(\SEst)
    = \Nbd(\SEst_t) \setminus \union_{j=0}^{t} \SEst_t;
$ %
i.e., it is the set of neighbors of $\SEst_t$ that we believe
to be uninfected at the beginning of cascade timestep $t$.

For a set of vertices $W \subseteq [n]$ and a vertex $v\in[n]$, let $\deg_W(v)$ denote the number of edges incident on $v$ that
are also incident on vertices in $W$.
We have, by linearity of expectation,
\begin{align}
\mu_t(\SEst) 
&= \sum_{v \in [n]} \Pr[\text{$v$ gets infected at time $t$} ~|~ \sigma_{t}(\SEst)] \\
&= \sum_{v \notin \Frontier_t}  \Pr[\text{$v$ gets infected at time $t$} ~|~ \sigma_{t}(\SEst)]
+ \sum_{v \in \Frontier_t }\Pr[\text{$v$ gets infected at time $t$} ~|~ \sigma_{t}(\SEst)] 
\\
&= \sum_{v\notin \Frontier_t \lunion \union_{j=0}^t \SEst_j}  \Pr[\text{$v$ gets infected at time $t$} ~|~ \sigma_{t}(\SEst)] +  \sum_{v \in \Frontier_t }\Pr[\text{$v$ gets infected at time $t$} ~|~ \sigma_{t}(\SEst)] \\
&= p_e \cdot \left(n - |\Frontier_t| - \sum_{j=0}^t |\SEst_j|\right) + \sum_{v\in \Frontier_t} (p_e + (1-p_e)(1 - (1 - p_n)^{\deg_{\SEst_t}(v)})).
\label{expr:MuTFormula}    
\end{align}
A similar expression can be derived for the more general case
where transmission probabilities across edges may differ from each other.

The calculation of the summation $\sum_{j=0}^t |\SEst_j|$ can be 
performed efficiently by keeping track of its value in the
$t$-th iteration of the loop of the algorithm.  In the $t$-th
iteration, the value of the summation is updated by adding
$|\SEst_t|$ to the running total. Note also that this estimation will be the only difference in our algorithm when applied to alternative cascade models such as the linear threshold model. 

\subsection{Approximation guarantee for FastClock}
Our first theorem gives an approximation guarantee for FastClock in the
case of the IC model.  It is
subject to a few assumptions about the temporal distortion model, the parameters of the cascade model and those of the graph model from which $G$ is sampled, which we state next. It is, however, important to note that FastClock itself does not assume anything about the graph. 

\begin{assumption}[Assumption on the temporal distortion model]
    \label{asmpt:observation-rate}
    No observed infection set $\hat{S}_i$ has too few vertices compared to
    the ground truth infection set $S_j$ from which it came.  In particular, this
    means that, for all $i$, the width of each observation interval $C_i$
    is bounded above by a constant, and there is some
    absolute constant $\epsilon > 0$ such that,
    for each $j \in C_i$, 
    \begin{align}
        |\hat{S}_j| \geq \epsilon \sum_{k \in C_i} |\hat{S}_k|.
    \end{align}
    Note that
    we do not assume anything else about the distribution of vertices in
    these observation intervals.  Furthermore, this assumption can be 
    somewhat relaxed to hold with high enough probability.
\end{assumption}

\begin{assumption}[Assumptions on random graph model parameters]
    \label{asmpt:model-params}
    We assume that $G \sim G(n, p)$ (i.e., that $G$ is sampled from the Erd\H{o}s-R\'enyi model), where $p$ satisfies the following
    relation with the ground truth number of cascade timesteps $T$:
    $ %
        p = o(n^{-\frac{T}{T+1}})
    $ %
    and
    $ %
    p \geq C \log n / n,
    $ %
    for some $C > 1$.
    The former condition may be viewed as a constraint on $T$, for
    a given choice of $p$.  It is natural in light of the fact that,
    together with our assumptions on $p_n$ and $p_e$ below,
    it implies that the cascade does not flood the graph, in the
    sense of infecting a $\Theta(1)$-fraction of nodes.  Many cascades
    in practice do not flood the graph in this sense.
    
    The latter condition implies that the graph is connected with
    high probability.

    Regarding parameters of the IC process, 
    we assume that $p_n$ is some fixed positive constant and
    that $p_e = o(p)$.  Our results also hold if $p_n$ is different
    for every edge $e$ (so that $p_n = p_n(e)$), provided that there
    are two positive constants $0 < c_0, c_1 < 1$ such that for
    every edge $e$, $p_n(e) \in [c_0, c_1]$.
    
    The assumption that $p_n$ is constant with respect to $n$ is 
    natural in the sense that, for many infectious processes,
    the probability of transmission from one node to another
    should not depend on the number of nodes.

    The assumption on $p_e$, the probability of infection from
    an external source, is reasonable when the cascade is overwhelmingly
    driven by network effects, rather than external sources.
\end{assumption}

\begin{theorem}[Main FastClock approximation theorem]
    Suppose that Assumptions~\ref{asmpt:observation-rate} and
    ~\ref{asmpt:model-params} hold.
    We have, with probability at least $1 - e^{-\Omega(np)}$,
    \begin{align}
        d_{\hat{S}}(C, \hat{C}) = O( (np)^{-1/3} ), %
    \end{align}
    where we recall that $S$ is an infection sequence generated
    by the cascade model, $C$ is the ground truth clock, $\hat{S}$
    is the observed infection sequence generated by the temporal
    distortion model, and
    $\hat{C}$ is the output of FastClock.
    \label{thm:fastclock-approx-main}
\end{theorem}

\subsection{Running time analysis}

We have a strong guarantee on the running time of FastClock in the independent cascade
case.

\begin{theorem}[Running time of FastClock]
    The FastClock algorithm for the independent cascade model runs in time $O(N + n + m)$,
    where $m$ is the number of edges in the input graph.
    \label{thm:fastclock-running-time}
\end{theorem}

Thus, the running time of FastClock is asymptotically much smaller than that of
the dynamic programming estimator 
from~\cite{ditursi2017network}.

\section{Proof sketches}
\label{sec:proof-sketches}

In this section, we primarily give proof sketches (except where subsection
headers indicate full proofs).  Full proofs of all results are given in
the appendix.

\subsection{Sketch of proof of Theorem~\ref{thm:fastclock-approx-main}}
The proof of Theorem~\ref{thm:fastclock-approx-main} employs an auxiliary result (Theorem~\ref{thm:fastclock-utility} in the appendix, 
which we call the \emph{FastClock utility theorem})
stating that with high probability, for every $i$, the intersection
of the ground truth infection sequence element $S_i$ with the estimated
infection sequence element $\SEst_i$ %
is asymptotically equivalent
in cardinality to $S_i$ itself.  We prove this theorem by induction on
the infection sequence element index $i$, which requires a careful design
of the inductive hypothesis.

Given the utility theorem, the required upper bound on the distortion
$d_{\hat{S}}(C, \hat{C})$ follows by summing over 
all possible pairs $S_i, S_j$ of infection sequence elements in $S$, the 
ground truth infection sequence, then summing over all vertex pairs 
$u \in S_i, v \in S_j$.  This inner sum can be approximated using the 
utility theorem.

\subsection{Full proof of Theorem~\ref{thm:fastclock-running-time}}
We analyze the worst-case running time of FastClock as follows:
initialization takes $O(1)$ time.  The dominant contribution to the running time is
the \emph{while} loop.  Since $t_{obs}$ is initially $0$ and increases by at least
$1$ in each iteration, the total number of iterations is at most $N$.  The remaining
analysis involves showing that each vertex and edge is only processed, a constant number of 
times, in $O(1)$ of these loop iterations, so that the running time is at most 
$O(N + n + m)$, as claimed.

In particular, the calculation of $\mu_t$ in every step involves a summation over all
edges from currently active vertices to their uninfected neighbors, along with a calculation involving the current number of uninfected vertices (which we can keep track of using $O(1)$
calculations per iteration of the loop).  A vertex is only active
in a single iteration of the loop.  Thus, each of these edges is only processed once in
this step.
The calculation of $t'_{obs}$ entails calculating a sum over elements of $\hat{S}$
that are only processed once in all of the iterations of the loop.  The calculation of
all of the $|\hat{S}_i|$ can be done as a preprocessing step via an iteration over all 
$n$ vertices of $G$.  Finally, the calculation of $\Frontier_{t+1}$ entails a union
over the same set of elements of $\hat{S}$ as in the calculation of the maximum, followed
by a traversal of all edges incident on elements of $\SEst_{t+1}$ whose other ends connect
to uninfected vertices.  These operations involve
processing the vertices in $\SEst_{t+1}$ (which happens only in a single iteration of
the loop, and, thus, with the preprocessing step of calculating the $|\hat{S}|_i|$, only
a constant number of times in the entire algorithm).  The edges leading to elements of
$\Frontier_{t+1}$ from elements of $\SEst_{t+1}$ are traversed at most twice in the loop:
once in the building of $\Frontier_{t+1}$ and once in the next iteration in the calculation
of $\mu_t$.

This implies that each vertex and edge is only processed $O(1)$ times in the entire algorithm.
This leads to the claimed running time of $O(N+n+m)$, which completes the proof.
\section{Empirical results on synthetic graphs}
\label{sec:empirical}

In this section, we present empirical results on synthetic graphs and cascades. Our goal is to confirm the theoretical guarantees of \FastClock~and compare it to the \emph{dynamic programming (DP)} algorithm optimizing a proxy of the maximum likelihood for observed cascades proposed by~\cite{ditursi2017network}. Our comparative analysis focuses on (i) \emph{distance}  of the estimated clock from the
ground truth clock (see Definition~\ref{def:dist}) and (ii) empirical \emph{running time} of both techniques.

We generate synthetic graphs using both the Erd\H{o}s–R\'enyi and the Stochastic Block Model. We then generate  synthetic cascades on each graph using the independent cascade (IC) model. We employ the obtained cascade sequence $S$ as the ground truth infection sequence, and create corresponding distorted (disaggregated) sequences $\hat{S}$ by ``stretching'' each ground truth time step of $S$. Specifically, to obtain a sample of a distorted observation sequence $\hat{S}$, we distribute the activated nodes in the ground truth time steps to $l$ corresponding time steps, where each node is placed in one of these $l$ timesteps uniformly at random.  Here, $l$ is an integer \emph{stretch} factor greater than $1$.  This implicitly specifies a clock $C$ on the stretched timeline, which we would like to infer (we note
that while all of our experiments involve a uniformly stretched timeline, our theoretical
contributions are more general).
We then employ both \FastClock~ and DP to estimate the ground truth clock from $\hat{S}$. We draw $50$ samples for each setting and report average and standard deviation for both running time and quality of estimations for each setting.

\paragraph{Experiments on Erd\H{o}s–R\'enyi graphs.} We first experiment with Erd\H{o}s–R\'enyi to confirm the theoretical behavior of our estimator and compare its running time and quality to the DP baseline. We report the results in Figure~\ref{fig:er}. With increasing graph size \FastClock's distance from the ground truth clock diminishes (as expected based on Theorem~\ref{thm:fastclock-approx-main}), while that of DP increases (Fig.~\ref{fig:q-n}). Note that DP optimizes a proxy to the cascade likelihood and in our experiments tend to over-aggregate the timeline which for large graph sizes results in incorrect recovery of the ground truth clocks. Similarly, \FastClock's estimate quality is better than that of DP for varying on $p_n$ (Fig.~\ref{fig:q-pn}), graph density (Fig.~\ref{fig:q-p}) and stretch factor for the cascades (Fig.~\ref{fig:q-stretch}), with distance from ground truth close to $0$ for regimes aligned with the key assumptions we make for our main results (Assumptions~\ref{asmpt:observation-rate},\ref{asmpt:model-params}). In addition to superior quality, \FastClock's running time scales linearly with the graph size and is orders of magnitude smaller than that of DP for sufficiently large instances (Figs.~\ref{fig:t-n},~\ref{fig:t-pn},~\ref{fig:t-p},~\ref{fig:t-stretch}).

\begin{figure*}[t!]
\footnotesize
\centering
\subfigure[Distance with $n$][]
 { \centering
  \includegraphics[width=0.23\textwidth]{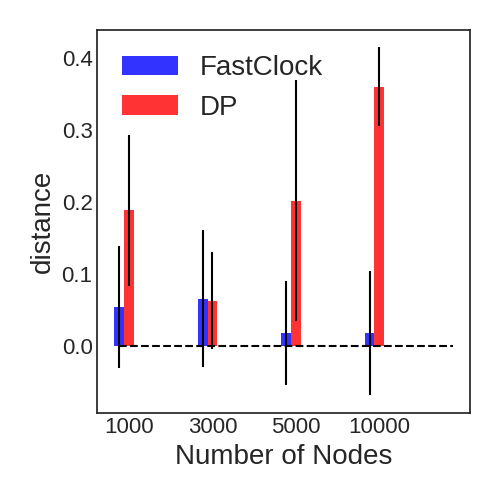}
  \label{fig:q-n}
 }
 \subfigure[Run time with $n$][]
 { \centering
  \includegraphics[width=0.23\textwidth]{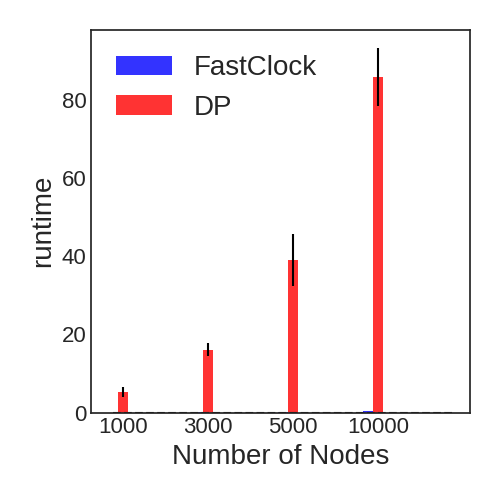}
  \label{fig:t-n}
 }
 \subfigure[Distance with $p_n$][]
 { \centering
  \includegraphics[width=0.23\textwidth]{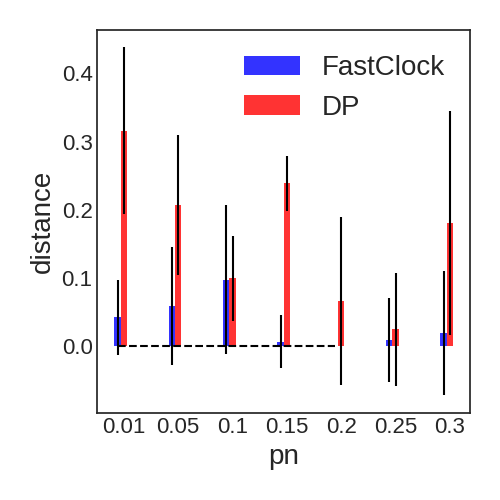}
  \label{fig:q-pn}
  }
  \subfigure[Run time with $p_n$][]
  { \centering
  \includegraphics[width=0.23\textwidth]{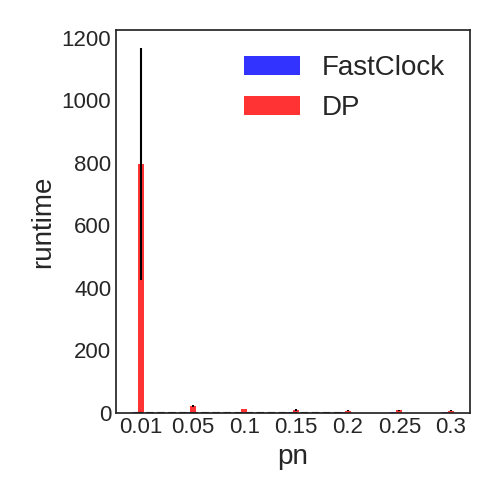}
  \label{fig:t-pn}
  }
  \subfigure[Distance with density][]
 { \centering
  \includegraphics[width=0.23\textwidth]{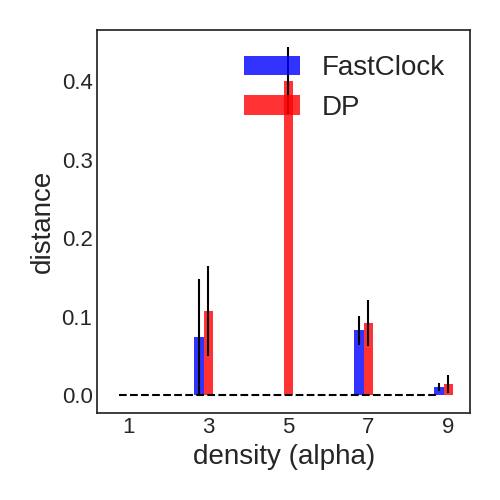}
  \label{fig:q-p}
 }
 \subfigure[Run time with density][]
 { \centering
  \includegraphics[width=0.23\textwidth]{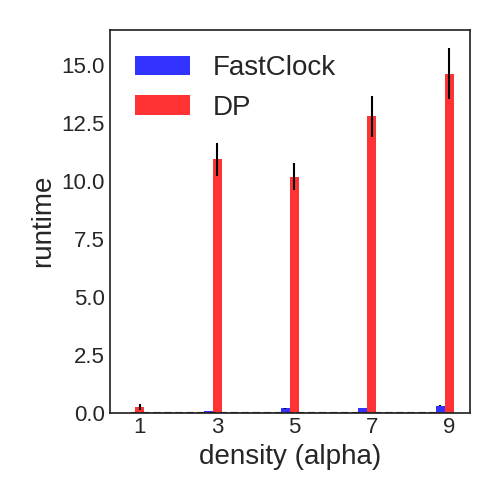}
  \label{fig:t-p}
 }
 \subfigure[Distance with stretch][]
 { \centering
   \includegraphics[width=0.23\textwidth]{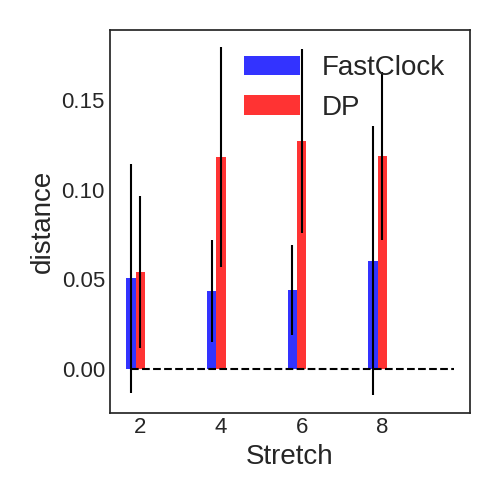}
  \label{fig:q-stretch}
  }
  \subfigure[Run time with stretch][]
  { \centering
  \includegraphics[width=0.23\textwidth]{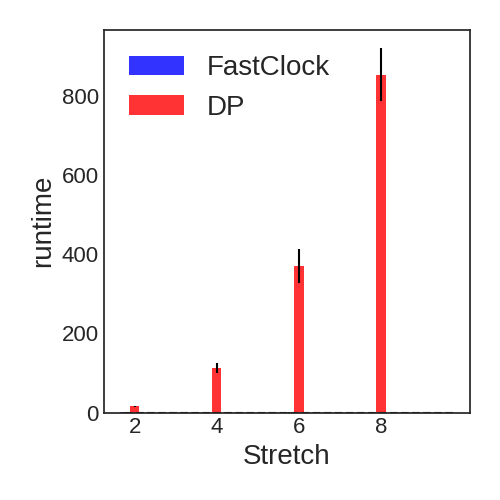}
  \label{fig:t-stretch}
  }
\caption{\footnotesize Comparison of the distance and runtime of the estimated clocks by \FastClock~ and the baseline DP from \cite{ditursi2017network} on Erd\H{o}s–R\'enyi graphs (default parameters for all experiments: $p_n=0.1$, $p_e=10^{-7}$, $n=3000$, $p=n^{-1/3}$, stretch $l=2$ unless varying in the specific experiment). (a),(b): Varying graph size.  (c),(d): Varying infection probability $p_n$. (e),(f): Varying graph density $p=n^{-1/\alpha}$. (g),(h): Varying stretch. \vspace{-0.4in}
}
\label{fig:er}
\end{figure*}

\paragraph{Experiments on Stochastic Block Model (SBM) graphs.} We would also like to understand the behavior of our estimator on graphs with communities where the cascade may cross community boundaries. To this end, we experiment with SBM graphs varying the inter-block connectivity and virality ($p_n$) of the cascades and report results in Fig.~\ref{fig:sbm}. As the cross-block connectivity increases and approaches that within blocks (i.e. the graph structure approaches ER-graph) \FastClock's quality improves and is significantly better than that of DP (Fig.~\ref{fig:bm-q-conn}). When, however, the transmission probability $p_n$ is high, coupled with sparse inter-block connectivity, \FastClock's estimation quality deteriorates beyond that of DP (Fig.~\ref{fig:bm-q-pn}). This behavior is due to the relatively large variance of $\mu_t$ when the cascade crosses a sparse cut in the graph with high probability. This challenging scenario opens an important research direction we plan to explore in future work. 

\begin{figure*}[t!]
\footnotesize
\centering
\subfigure[Dist. with connectivity][]
 { \centering
  \includegraphics[width=0.23\textwidth]{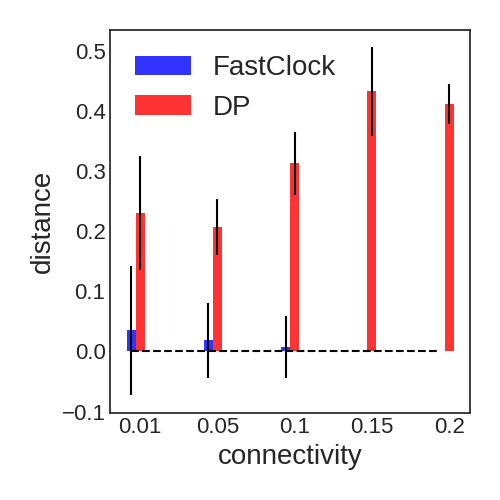}
  \label{fig:bm-q-conn}
 }
 \subfigure[Time with connectivity][]
 { \centering
  \includegraphics[width=0.23\textwidth]{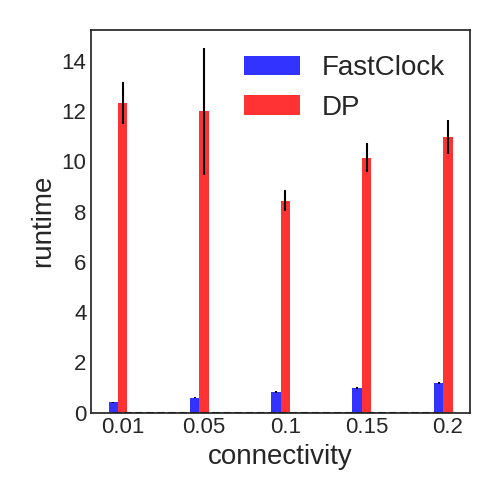}
  \label{fig:bm-t-conn}
 }
 \subfigure[Distance with $p_n$][]
 { \centering
  \includegraphics[width=0.23\textwidth]{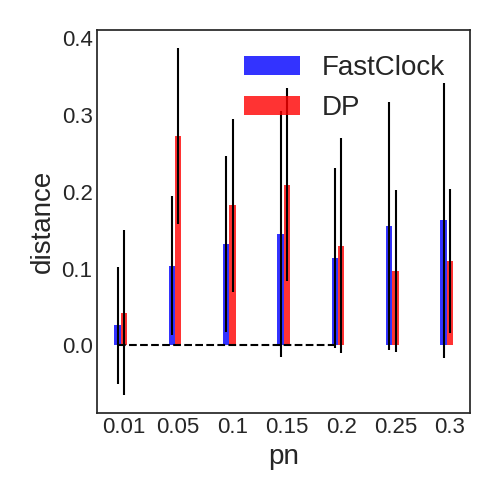}
  \label{fig:bm-q-pn}
  }
  \subfigure[Run time with $p_n$][]
  { \centering
  \includegraphics[width=0.23\textwidth]{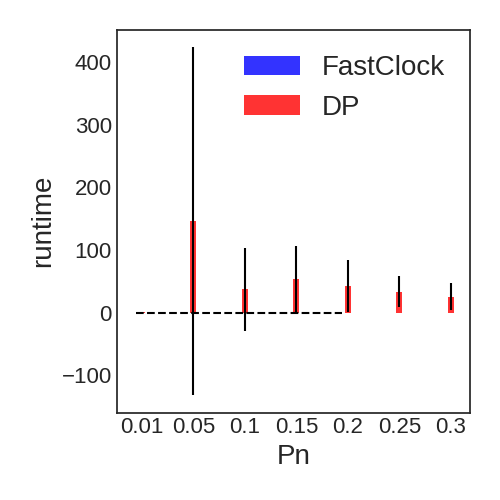}
  \label{fig:bm-t-pn}
  }
\caption{\footnotesize Comparison of the distance and runtime of the estimated clocks by \FastClock~ and the baseline DP from \cite{ditursi2017network} on Stochastic Block Model graphs (default parameters: $n=5000$, two blocks/communities of sizes $n/\sqrt{n}$ and $n-n/\sqrt{n}$, $p_e=10^{-7}$, stretch $l=2$). (a),(b): Varying inter-block connectivity ($p_n=0.1$) where a setting of $0.2$ makes the graph equivalent to an Erd\H{o}s–R\'enyi graph with $p=0.2$.  (c),(d): Varying infection probability $p_n$ (inter-block connectivity is set to $0.01$). %
}
\label{fig:sbm}
\vspace{-0.3in}
\end{figure*}

\section{Conclusions and future work}
\label{sec:conclusions}
We have formulated a statistical estimation framework for the problem of recovery of all states of a discrete-time cascade from temporally distorted
observation sequences.  In the case of oversampling clocks, we showed that
temporal distortion can be corrected with high accuracy and low computational
cost, subject to certain natural constraints on the structure of the underlying
graph and on the cascade model: in essence, these must be such that
the graph is an expander with appropriate parameters; that, conditioned on an 
estimated current state of the process at any time, the expected number of 
vertices infected in the next timestep is immune to small errors in the 
estimated state;
and that the number of vertices infected in the next timestep is 
well-concentrated around its conditional expected value.  
We empirically showed that the FastClock algorithm is superior in accuracy and running time
to the current state of the art dynamic programming algorithm.  Furthermore, unlike this baseline, FastClock comes with theoretical accuracy guarantees.
Our results
are formally stated for the independent cascade model, but they very
likely hold for a broad class of other models, including the linear threshold 
model.

We intend to pursue further work on this problem: most pressingly, 
our empirical results and intuition derived from our theorems indicate
that FastClock may not perform accurately when the graph contains
very sparse cuts (so that it is not an expander graph).  Further work
is needed to determine whether accuracy and computational efficiency
can be achieved for such graphs.
Furthermore, our method relies on knowledge of the parameters of the
cascade process.  We intend to investigate the extent to which this assumption can be relaxed.

{
\bibliography{references}
}
\appendix

\section{Glossary of notation}

Here we collect the notation that is used in the main body of the paper and in the proofs in the appendix.

\begin{enumerate}
    \item
        $\Nbd(S)$: Neighborhood of the set $S$ of vertices in a given graph.
    \item 
        $S = (S_0, S_1, ..., S_T)$ -- An infection sequence generated
        by a cascade model with $T+1$ timesteps.  Each $S_j$ is a subset
        of vertices, and $S_i \lintersect S_j = \emptyset$ for $i\neq j$.
        We denote by $|S|$ the number of timesteps of $S$: $T+1$.
        We think of $S$ as the ground truth infection sequence.
    \item
        $\hat{S} = (\hat{S}_0, \hat{S}_1, ..., \hat{S}_N)$ -- An observation
        of an infection sequence that has been temporally distorted by a clock.
    \item
        $C$ -- The ground-truth clock in our estimation problem.
    \item
        $\hat{C}$ -- The clock estimated by our algorithm.
    \item
        $\SEst$ -- The estimate of the original infection sequence $S$
        induced by our estimate $\hat{C}$ of the clock $C$ applied to
        the observed infection sequence $\hat{S}$.
    \item
        $\sigma_{t}(S)$, for an infection sequence $S$ and a timestep index 
        $t \in |S|$ -- 
        The $\sigma$-field generated by the event that the first $t$ 
        infection sets of the IC process are given by $S_0, ..., S_t$.
    \item
        $\mu_t(S)$, for an infection sequence $S$ and a timestep index
        $t \in |S|$ -- $\E[|S_{t+1}| ~|~ \sigma_t(S)]$.  This is the
        expected number of vertices infected in the $t+1$st timestep,
        given the infection sequence up to and including timestep $t$.
    \item
        $N$ -- The index of the last observed infection set.  That is, $|\hat{S}| = N+1$.
    \item
        $T$ -- The index of the last ground truth infection set.  
        That is, $|S| = T+1$.
    \item
        $n$ -- The size of the graph.
    \item
        $p_n$ -- The probability in the IC model of transmission across an edge in a single
        timestep.
    \item
        $p_e$ -- The probability of infection of a vertex in a single timestep by a non-network source.
    \item
        $R(S, i)$ -- For an infection sequence $S$ and an index $i$, define 
        the $i$th \emph{running sum}
        to be
        \begin{align}
            R(S, i)
            = \union_{j\leq i} S_j.
        \end{align}
    \item
        $\Frontier(S, i)$ -- For an infection sequence $S$ and an index $i \in \{0, 1, ..., |S|\}$, define the $i$th frontier set to be 
        \begin{align}
            \Frontier(S, i)
            = \Nbd(S_i) \setminus R(S, i).
        \end{align}
        The $i$th frontier with respect to $S$ is the set of neighbors
        of vertices infected in timestep $i$ that have not infected
        by the end of timestep $i$.
    \item
        $\CF(S, i)$ -- The candidate frontier set at the end of timestep
        $i$ in infection sequence $S$.  That is, this is 
        \begin{align}
            \CF(S, i)
            = [n] \setminus R(S, i).
        \end{align}
        Note that $\Frontier(S, i) \subseteq \CF(S, i)$.
    \item
        $\CCF(S, \SEst, i, j)$ -- The common candidate frontier:
        \begin{align}
            \CCF(S, \SEst, i, j)
            = \CF(S, i) \lintersect \CF(\SEst, j).
        \end{align}

\end{enumerate}

\section{Proofs}

In this section, we give full proofs of all results.

\subsection{Proof of Theorem~\ref{thm:fastclock-approx-main}}

To prove the main FastClock approximation theorem, we start by characterizing the growth of $\mu_{i}(S)$ and $|S_i|$
as a function of $n$ and $i$.
Note that this is a result about the independent cascade process, not
the FastClock algorithm.
\begin{lemma}[Growth of $\mu_i(S)$ and $|S_i|$]
    We have that, with probability at least $1 - e^{-np}$, for all $i \leq T$, 
    \begin{align}
        \mu_i(S) = \Theta((np)^{i+1}),
    \end{align}
    where the $\Theta(\cdot)$ is uniform in $i$.
    Furthermore, with probability at least $1 - e^{-np}$,
    we have
    \begin{align}
        |S_i| = \Theta((np)^{i})
    \end{align}
    for every $i$.
    \label{lemma:mu_i-growth}
\end{lemma}
\begin{proof}
    We prove this by induction on $i$ and use the formula (\ref{expr:MuTFormula})
    throughout.
    
    \paragraph{Base case ($i = 0$):}
    In the base case, we are to verify that $\mu_0(S) = \Theta(np)$.  The first
    term of (\ref{expr:MuTFormula}) is non-negative and at most $p_e \cdot n$.  By
    our assumption, we have that $p_e = o(p_n)$, implying that the first term
    is $o(np)$.  Thus, it remains for us to show that the second sum is $\Theta(np)$.
    The dominant contribution comes from the second term of each term of the sum:
    \begin{align}
        \sum_{v \in \Frontier_0(S)} (p_e + (1-p_e)(1-  (1-p_n))^{\deg_{S_0}(v)})
        &= \Theta(\sum_{v \in \Frontier_0(S)} 1 - (1-p_n)^{ \deg_{S_0}(v) } ) \\
        &= \Theta( |\Frontier_0(S)| + \sum_{v\in\Frontier_0(S)}(1-p_n)^{ \deg_{S_0}(v) } ).
    \end{align}
    In the final expression above, the remaining sum is lower bounded by $0$ and upper bounded
    by $|\Frontier_0(S)|$, since each term is between $0$ and $1$.  Thus, we have shown that,
    with probability exactly $1$,
    \begin{align}
        \mu_0(S)
        = \Theta(|\Frontier_0(S)|) + o(np).
    \end{align}
    Since $|\Frontier_0(S)|$ is the set of uninfected neighbors of all vertices in $S_0$,
    and, by assumption, $|S_0| = \Theta(1)$, we have that with probability
    at least $1 - e^{-np}$,
    \begin{align}
        |\Frontier_0(S)| = \Theta(np).
    \end{align}
    Thus, we have
    \begin{align}
        \mu_0(S) = \Theta(np)
    \end{align}
    with probability $\geq 1-e^{-np}$. 
    Conditioning on this event (which is only an event dealing with the
    graph structure), we have that $|S_1| \sim \Binomial(\Theta(np), p_n)$,
    and a Chernoff bound gives us that with probability $1-e^{-\Omega(np)}$, $|S_1| = \Theta(np)$, as desired.
    This completes the proof of the base case.
    
    \paragraph{Induction ($i > 0$, and we verify the inductive hypothesis for $i$):}
    We assume that $\mu_{j}(S) = \Theta((np)^{j+1})$ and $|S_{j+1}| = \Theta((np)^{j+1})$ for 
    $j = 0, 1, ..., i-1$.  We must verify that it holds for $j=i$, with 
    probability at least $1 - e^{-np}$.  As in the base case, the first 
    term of (\ref{expr:MuTFormula}) is $O(np_e) \ll np \ll (np)^{i+1}$.
    It is, therefore, negligible with probability $1$.  The second
    term again provides the dominant contribution and is easily seen
    to be $\Theta(|\Frontier_i(S)|)$, just as in the base case.
    Thus, it remains to show that $|\Frontier_i(S)| = \Theta((np)^{i+1})$
    with probability at least $1 - e^{-\Omega(np)}$,
    which implies the desired result for $\mu_i(S)$.  The inductive hypothesis
    implies that $|S_{i}| = \Theta((np)^{i})$, and the number of uninfected
    vertices is $n - \sum_{j=0}^{i} |S_j| = n - \Theta((np)^{i+1})$.
    Since $i \leq T-1$, this is asymptotically equivalent to $n$.
    
    Now, conditioned on the first $i$ elements of $S$, the $i$th
    frontier $|\Frontier_i(S)| \sim \Binomial(n\cdot(1-o(1)), 1-(1-p)^{|S_i|})$.  Thus, with probability at least
    $1 - e^{-\Omega((np)^{i})}$, we have
    \begin{align}
        |\Frontier_i(S)|
        = \Theta( n \cdot (1 - (1-p)^{|S_i|})).
    \end{align}
    Now, 
    \begin{align}
        1 - (1-p)^{|S_i|}
        = 1 - (1-p)^{(np)^{i}}.
    \end{align}
    Since $p = o(1)$, we have
    \begin{align}
        1 - (1-p)^{ (np)^{i} }
        \sim
        1 - e^{-p^{i+1}n^{i}}
    \end{align}
    From (\ref{expr:pnPowerAsymptotics}) below, we have that
    \begin{align}
        p^{i+1}n^i = o(1).
    \end{align}
    This implies that
    \begin{align}
        1 - e^{-p^{i+1}n^{i}}
        = 1 - (1 - p^{i+1}n^i)(1 + O(p^{i+1}n^i))
        = p^{i+1}n^i (1 + o(1)).
    \end{align}
    Thus, with probability at least $1 - e^{-\Omega((np)^{i})}$,
    \begin{align}
        |\Frontier_i(S)|
        = \Theta( (np)^{i+1}),
    \end{align}
    which implies that
    \begin{align}
        \mu_i(S)
        = \Theta( (np)^{i+1}).
    \end{align}
    By concentration of $|S_i|$, we then have that with probability
    at least $1 - e^{-\Omega(\mu_i(S))}$, 
    \begin{align}
        |S_i| = \Theta((np)^{i+1}), 
    \end{align}
    as desired.
    
    \paragraph{Completing the proof}
    Let $G_i$ be the event that the inductive hypothesis holds
    for index $i = 0, 1, ..., T-1$.  Then we have
    \begin{align}
        \Pr[  \lintersect_{i\geq 0} G_i]
        = \Pr[G_0] \cdot \prod_{i\geq 1} \Pr[ G_i ~|~ \lintersect_{j=0}^{i-1} G_j]
        \geq \prod_{i=0}^{T-1} (1 - e^{-\Omega((np)^i)})
        1 - e^{-\Omega((np))}.
    \end{align}
    This completes the proof.
\end{proof}

Next, we state
and prove a utility theorem (Theorem~\ref{thm:fastclock-utility} below).
To state it, we need some notation: our estimated clock $\hat{C}$ 
induces an estimate $\SEst$ of the ground truth infection sequence
$S$.  In particular, $\SEst$ is the unique infection sequence such
that distorting $\SEst$ according to $\hat{C}$ yields $\hat{S}$
as an observed infection sequence.

\begin{theorem}[Main FastClock analysis utility theorem]
    We have that with probability $1 - e^{-\Omega(np)}$,
    for every $i \leq T-1$, 
    \begin{align}
        |S_i \lintersect \SEst_i|
        = |S_i| \cdot (1 - O((np)^{-1/3})).
    \end{align}
    \label{thm:fastclock-utility}
\end{theorem}    

We will prove this theorem by induction on $i$.  The inductive hypothesis
needed is subtle, as a na\:ive hypothesis is too weak.  To formulate it and
to prove our result,
we need some notation: for an infection sequence $W$, we define the $i$th
running sum to be 
\begin{align}
    R(W, i)
    = \union_{j=0}^i W_j.
\end{align}
We define the \emph{frontier} and \emph{running sum discrepancy sets}
as follows:
\begin{align}
    \FrontierDisc(S, \hat{S}, i, j)
    = \Frontier_i(S) ~\symmdiff~ \Frontier_j(\hat{S}) \\
    \RDisc(S, \hat{S}, i, j)
    = R(S, i) ~\symmdiff~ R(\hat{S}, j),
\end{align}
where $\symmdiff$ denotes the symmetric difference between two sets.

We define the \emph{candidate frontier} at timestep $i$ in infection sequence $S$ to be
\begin{align}
    \CF(S, i)
    = [n] \setminus R(S, i).
\end{align}
This is the set of vertices that are not yet infected after timestep $i$.

We define the \emph{common candidate frontier} to be
\begin{align}
    \CCF(S, \hat{S}, i, j)
    = \CF(S, i) \lintersect \CF(\hat{S}, j).
\end{align}

With this notation in hand, we define the following inductive hypotheses:
\begin{enumerate}
    \item 
        There is a small discrepancy between the running sums of the true and estimated
        clocks:
        \begin{align}
            ||R(S, i)| - |R(\SEst, i||
            \leq f_1(n, i),
            \label{expr:R_inductive}
        \end{align}
        where we
        set, with foresight, $f_1(n, i) = \mu_{i-1}(S)^{.66} = o(\mu_{i-1}(S)^{2/3})$.
    \item
        There is a small discrepancy between $S_i$ and $\SEst_{i}$:
        \begin{align}
            1 - \frac{| S_i \lintersect \SEst_{i}|}{|S_i|} 
            \leq f_2(n, i),
            \label{expr:C_inductive}
        \end{align}
        where we set, with foresight, $f_2(n, i) = D\cdot \mu_{i-1}(S)^{-1/3}$, for some large enough absolute constant $D$.
\end{enumerate}
We will use these to prove Theorem~\ref{thm:fastclock-utility}.  The
base case and inductive steps are proven in Propositions~\ref{prop:base-case} and~\ref{prop:inductive-step} below.
First, we
start by proving an upper bound (Theorem~\ref{thm:mu_diff_upper_bound}) on the following difference:
\begin{align}
    |\mu_i(S) - \mu_{i}(\SEst)|.
    \label{expr:mu_diff_upper_bound}
\end{align}
    In essence, the upper bound says that at any given clock time step, the expected number of
    nodes infected in the next timestep is almost the same according to both
    the true and estimated clock.
    This
    will later be used verify the two inductive
    hypotheses stated above.
    
    \begin{theorem}[Upper bound on (\ref{expr:mu_diff_upper_bound})]
        Granted the inductive hypotheses explained above, we have 
        that
        \begin{align}
            |\mu_{i}(S) - \mu_{i}(\SEst)|
            \leq p \mu_{i-1}(S)^{2/3} \mu_i(S),
        \end{align}
        with probability $\geq 1 - e^{-\Omega(\mu_{i}(S))}$.
        \label{thm:mu_diff_upper_bound}
    \end{theorem}
    \begin{proof}
        To upper bound (\ref{expr:mu_diff_upper_bound}), we apply the triangle inequality to (\ref{expr:MuTFormula})
        to get
        \begin{align}
            |\mu_i(S) - \mu_{i}(\SEst)|
            &\leq 
                p_e \cdot \left| |\Frontier_i(S)| - |\Frontier_{i}(\SEst)|  \right| 
                \label{expr:FrontierDiscTerm} \\
                &+ p_e \left| \sum_{j=0}^{i} |\SEst_j| - \sum_{j=0}^i |S_j| \right| \label{expr:RDiscTerm} \\
                &+ \left| \sum_{v\in \Frontier_i(S)} Q(i, S, v) - \sum_{v \in \Frontier_{i}(\SEst)} Q(i, \SEst, v) \right|, \label{expr:QDiscTerm}
        \end{align}
        where $Q(i, S, v) = p_e + (1-p_e)(1 - (1-p_n)^{\deg_{S_i}(v)})$.

        We will upper bound each of the three terms (\ref{expr:FrontierDiscTerm}), (\ref{expr:RDiscTerm}), and (\ref{expr:QDiscTerm}) separately.

        \paragraph{\textbf{Upper bounding (\ref{expr:FrontierDiscTerm}) by $O(p_e |\Frontier_i(S)|)$:} }
        
        We first note that 
        \begin{align}
            \left| |\Frontier_i(S)| - |\Frontier_{i}(\SEst)|  \right| 
            \leq |\Delta\F(S, \SEst, i, i)|.
        \end{align}
        So it is enough to upper bound the frontier discrepancy set cardinality.
        In order to do this, we decompose it as follows:
        \begin{align}
            |\Delta\Frontier(S, \SEst, i, i)| 
            = |\Delta\Frontier(S, \SEst, i, i) \lintersect \Delta R(S, \SEst, i, i)|
              +  |\Delta\Frontier(S, \SEst, i, i) \lintersect \CCF(S, \SEst, i, i)|.
            \label{expr:DeltaFrontierDecomposition}      
        \end{align}
        This decomposition holds for the following reason:
        let $v$ be a vertex in
        the frontier discrepancy set $\Delta\Frontier(S, \SEst, i, i)$.  Suppose, further, that $v$ is not
        in the common candidate frontier for $S_i, \SEst_i$ (so it does not contribute to the second term on the right-hand side of (\ref{expr:DeltaFrontierDecomposition})).  We will show that it must
        be a member of $\Delta R(S, \SEst, i, i)$, which will complete the proof of the claimed decomposition.  Then 
        $v$ must be a member of at least one of $R(S, i), R(\SEst, i)$ (i.e., it must already be infected in at least one of these).  If it were a member of both, then it would not be a member of either frontier, so it could
        not be a member of the frontier discrepancy set.  Thus, it $v$ is
        only a member of one of $R(S, i)$ or $R(\SEst, i)$.  This implies
        that $v \in \Delta R(S, \SEst, i, i)$.  This directly implies
        the claimed decomposition (\ref{expr:DeltaFrontierDecomposition}).
        
        We now compute the expected value of each term of 
        the right-hand side of (\ref{expr:DeltaFrontierDecomposition}), where the expectation is taken with respect to the graph $G$.  After upper bounding the 
        expectations, standard concentration inequalities will complete our 
        claimed bound on the size of the frontier discrepancy set.  
        
        In the first 
        term, the size of the intersection of the frontier discrepancy with the 
        running sum discrepancy is simply
        the number of vertices in the running sum discrepancy set that have at least one edge
        to some vertex in $S_i$ (here we assume, without loss of generality, that
        $|R(S, i)| \leq |R(\SEst, i)|$).  Using linearity of expectation, 
        the expected number of such vertices is
        \begin{align}
            \E[|\Delta\Frontier(S, \SEst, i, i) \lintersect \Delta R(S, \SEst, i, i)|]
            = |\Delta R(S, \SEst, i, i)| \cdot (1 - (1-p)^{|S_i|}).
            \label{expr:ExpectedDeltaFrontierTerm1}
        \end{align}
        Here $(1 - (1-p)^{|S_i|})$ is the probability that, for a fixed vertex 
        $w \in \Delta R(S, \SEst, i, i)$, there is at least one edge between $w$
        and some vertex in $S_i$.
        
        We compute the expected value of the second term of (\ref{expr:DeltaFrontierDecomposition}) as follows.
        
        We claim that 
        \begin{align}
            \Delta\Frontier(S, \SEst, i, i) \lintersect \CCF(S, \SEst, i, i)
            \subseteq \CCF(S, \SEst, i, i) \lintersect ( \Nbd(\Delta R(S, \SEst, i, i)) \setminus \Nbd(S_i) ).
        \end{align}
        To show this, let $v \in \Delta\Frontier(S, \SEst, i, i) \lintersect \CCF(S, \SEst, i, i)$. The fact that $v$ is in the frontier discrepancy set means that it has an edge to 
        exactly one of $S_i, \SEst_i$.  This implies that it has an edge to the running sum
        discrepancy set.  Recalling that we assumed wlog that $|R(S, i)| \leq |R(\SEst, i)|$, 
        we must have that $\Delta R(S, \SEst, i, i) \lintersect S_i = \emptyset$, and so we 
        must also have that there are no edges from $v$ to $S_i$.  This completes the proof
        of the claimed set inclusion.  This implies that
        \begin{align}
            \E[|\Delta\Frontier(S, \SEst, i, i) \lintersect \CCF(S, \SEst, i, i)|]
            \leq \E[|\CCF(S, \SEst, i, i) \lintersect ( \Nbd(\Delta R(S, \SEst, i, i)) \setminus \Nbd(S_i) )|].
        \end{align}
        As above, the expectation is taken with respect to the random graph $G$.
        
        For a single vertex in the common candidate frontier, the probability that it lies 
        in the frontier discrepancy set is thus at most
        \begin{align}
            (1 -  (1-p)^{|\Delta R(S, \SEst, i, i)|}) \cdot (1-p)^{|S_i|}.
        \end{align}
        Thus, using linearity of expectation, the expected size of the 
        second term in (\ref{expr:DeltaFrontierDecomposition}) is upper bounded by 
        \begin{align}
            \E[|\Delta\Frontier(S, \SEst, i, i) \lintersect \CCF(S, \SEst, i, i)| ~|~ \sigma_i(S)]
            \leq |\CCF(S, \SEst, i, i)| \cdot  (1 -  (1-p)^{|\Delta R(S, \SEst, i, i)|}) \cdot (1-p)^{|S_i|}.
            \label{expr:ExpectedDeltaFrontierTerm2}
        \end{align}
        Combining (\ref{expr:ExpectedDeltaFrontierTerm1}) and 
        (\ref{expr:ExpectedDeltaFrontierTerm2}) and defining $q = 1-p$, 
        we have the following expression for the expected size of the frontier discrepancy
        set:
        \begin{align}
            &\E[ |\Delta\Frontier(S, \SEst, i, i)|] \label{expr:DeltaFExpected} \\
            &= |\Delta R(S, \SEst, i, i)| \cdot (1 - q^{|S_i|}) \\
              &\;\;+ |\CCF(S, \SEst, i, i)| \cdot  (1 -  q^{|\Delta R(S, \SEst, i, i)|}) \cdot q^{|S_i|}.
              \label{expr:DeltaFExpectedCCF}
        \end{align}

        We would like this to be $O(\E[|\Frontier_i(S)| ~|~ \sigma_i(S)])$.
        Note that $\E[|\Frontier_i(S)| ~|~ \sigma_i(S)]$ can be expressed as follows:
        \begin{align}
            \E[|\Frontier_i(S)| ~|~ \sigma_i(S)]
            &= (|\Delta R(S, \SEst, i, i)| \label{expr:FExpected}  \\
              &\;\;+ |\CCF(S, \SEst, i, i)|) \cdot  (1 -  q^{|S_i|}).
        \end{align}
        
        The intuition behind (\ref{expr:DeltaFExpected}) being $O(\E[|\Frontier_i(S)| ~|~ \sigma_i(S)])$ is as follows:
        the $\Delta R$ term is exactly the same as in (\ref{expr:FExpected}).  However,
        this term is negligible compared to the common candidate frontier term in both
        expected values.  The second term, (\ref{expr:DeltaFExpectedCCF}), can be asymptotically
        simplified as follows: we have
        \begin{align}
            1 - q^{|\Delta R(S, \hat{S}, i, i)|}
            &= 1 - (1 - p)^{ |\Delta R(S, \SEst, i, i)| } \\ 
            &\sim 1 - (1 - p)\cdot  |\Delta R(S, \SEst, i, i)|) \\
            &= p\cdot  |\Delta R(S, \SEst, i, i)| \\ 
            &= p |S_i| \cdot \frac{|\Delta R(S, \SEst, i, i)| }{|S_i|} \\
            &\sim (1 - q^{|S_i|}) \cdot  \frac{|\Delta R(S, \SEst, i, i)| }{|S_i|}.
        \end{align}
        Here, we have used the following facts:
        \begin{itemize}
            \item
                For the first asymptotic equivalence, we used the fact that
                $p |\Delta R(S, \SEst, i, i)|  = o(1)$.
                More precisely, we have from the inductive hypothesis that
                \begin{align}
                    |\Delta R(S, \SEst, i, i)|
                    = o(\mu_{i-1}(S)^{0.66})
                    = o( (np)^{i \cdot 0.66}),
                \end{align}
                so we have
                \begin{align}
                    p |\Delta R(S, \SEst, i, i)|
                    = o(p^{0.66i + 1} n^{0.66i})
                    = o(n^{-T/(T+1) + 0.66i/(T+1)}),
                \end{align}
                which is polynomially decaying in $n$.
            \item
                For the second asymptotic equivalence, we used the fact that
                $p|S_i| = o(1)$.  More precisely, this comes from the
                fact that
                \begin{align}
                    p|S_i|
                    = O(p (np)^i)
                    = O(p^{i+1}n^{i}).
                \end{align}
                Now, we use the fact that $p = o(n^{-\frac{T}{T+1}})$:
                \begin{align}
                    p^{i+1} n^{i}
                    = o(n^{-\frac{T}{T+1}(i+1) + i}),
                    \label{expr:pnPowerAsymptotics}
                \end{align}
                from our assumption on the growth of $p$.
                Now, we need to show that the exponent is sufficiently negative and bounded
                away from $0$.
                \begin{align}
                    -\frac{T}{T+1}(i+1) + i
                    = \frac{-T\cdot (i+1) + i\cdot(T+1)}{T+1}
                    = \frac{-T + i}{T+1}
                    \leq \frac{-1}{T+1}.
                \end{align}
                We have used the fact that $i \leq T-1$.
                Now, the constraints that we imposed on $p$ imply that
                $T = o(\log n)$, so
                \begin{align}
                    n^{\frac{-1}{T+1}}
                    = e^{\frac{-\log n}{T+1}}
                    = o(1),
                \end{align}
                as desired, since the exponent tends to $-\infty$ as $n\to\infty$.
        \end{itemize}
        
        Let us be more precise about what we proved so far.
        We have
        \begin{align}
            \E[ |\Delta\Frontier(S, \SEst, i, i)| ~|~ \sigma_i(S)] 
            &\sim (1 - q^{|S_i|}) \cdot |\CCF(S, \SEst, i, i)|  
              \cdot \left( \frac{|\Delta R|}{|\CCF|} + \frac{|\Delta R|}{|S_i|} q^{|S_i|} \right). 
        \end{align}
        Meanwhile,
        \begin{align}
            \E[ |\Frontier_i(S)| ~|~ \sigma_i(S)]
            = (1 - q^{|S_i|}) \cdot |\CCF| \cdot \left(1 + \frac{|\Delta R|}{|\CCF|} \right). 
        \end{align}
        We have that
        \begin{align}
            \frac{\E[|\Delta\F_i| ~|~ \sigma_i(S)]}{\E[ |\F_i| ~|~ \sigma_i(S)]}
            \sim \frac{ \frac{|\Delta R|}{|\CCF|}  + \frac{|\Delta R|}{|S_i|}\cdot q^{|S_i|}}{ 1 + \frac{|\Delta R|}{|\CCF |} }.
        \end{align}
        This can be simplified as follows:
        \begin{align}
            \frac{\E[|\Delta\F_i| ~|~ \sigma_i(S)]}{\E[ |\F_i| ~|~ \sigma_i(S)]}
            \sim
            \frac{ \frac{|\Delta R|}{|\CCF|}  + \frac{|\Delta R|}{|S_i|}\cdot q^{|S_i|}}{ 1 + \frac{|\Delta R|}{|\CCF |} }
            =
            \frac{ |\Delta R| \cdot \left( 1 + \frac{|\CCF| }{|S_i| }q^{|S_i|} \right) }{ |\CCF| + |\Delta R| }.
        \end{align}
        This can be upper bounded as follows, by distributing in the numerator
        and upper bounding $|\Delta R|$ by $|\Delta R| + |\CCF|$ in the numerator of
        the resulting first term:
        \begin{align}
            \frac{ |\Delta R| \cdot \left( 1 + \frac{|\CCF| }{|S_i| }q^{|S_i|} \right) }{ |\CCF| + |\Delta R| }
            \leq
            1 + 
            \frac{ |\Delta R| \cdot |\CCF| q^{|S_i|}}{|S_i|(|\CCF| + |\Delta R|)}.
        \end{align}
        We can further upper bound by noticing that $|\CCF| + |\Delta R| \geq |\CCF|$,
        so
        \begin{align}
            \frac{\E[|\Delta\Frontier_i| ~|~ \sigma_i(S)]}{\E[ |\Frontier_i| ~|~ \sigma_i(S)]}
            \leq
            1 + 
            \frac{ |\Delta R| }{|S_i|}.       
        \end{align}
        Now, by our inductive hypothesis, we know that $|\Delta R|_i = o(\mu_{i-1}(S)^{0.66})$, and by concentration, we know that $|S_i| = \Theta(\mu_{i-1}(S))$.  Thus, we have
        \begin{align}
            \frac{\E[|\Delta\Frontier_i| ~|~ \sigma_i]}{\E[ |\Frontier_i| ~|~ \sigma_i]}
            \leq
            1 + 
            \frac{ |\Delta R| }{|S_i|} 
            = 1 + o(\mu_{i-1}(S)^{-(1-0.66)})
            = O(1).
        \end{align}       
        
        Thus, 
        \begin{align}
            \E[ |\Delta\Frontier(S, \SEst, i, i)| ~|~ \sigma_i(S)]
            = O( \E[ |\Frontier_i(S)| ~|~ \sigma_i(S)] ).
        \end{align}

        Now, remember that our goal is to show that $|\Delta\Frontier(C, \hat{C}, i, i)| = O(|\Frontier_i(C)|)$ with high probability, conditioned on $\sigma_i(S)$.  This follows
        from the expectation bound above and the fact that the size of the frontier
        in both clocks is binomially distributed, so that standard concentration bounds
        apply.  This results in the following: 
        \begin{align}
            p_e|\Delta \Frontier_i| = O(p_e |\Frontier_i|)
            \label{expr:FDiscTermFinalUB}
        \end{align}
        with conditional probability at least $1 - e^{-\Omega((np)^i)}$.

        \paragraph{\textbf{Upper bounding (\ref{expr:RDiscTerm}) by $o(p_e |R(S, i)|)$}: }
        
        To upper bound (\ref{expr:RDiscTerm}), we note that
        \begin{align}
            \sum_{j=0}^i |S_j| = |R(S, i)|,
        \end{align}
        and an analogous identity holds with $\SEst$ in place of $S$.  Moreover,
        \begin{align}
            \left| R(S, i) - R(\SEst, i) \right|
            = |\Delta R(S, \SEst, i, i)|.
        \end{align}
        Thus, we have
        \begin{align}
            (\ref{expr:RDiscTerm})
            = p_e |\Delta R(S, \SEst, i, i)|
            \leq p_e f_1(n, i),
        \end{align}
        where the inequality is by the inductive hypothesis.
        We want this to be $o(p_e \cdot |R(S, i)|)$, which means that we want
        $|\Delta R(S, \SEst, i, i)| = o(|R(S, i)|)$.  This is follows from the inductive
        hypothesis.
        In particular, we know that
        $|R(S, i)| \geq |S_i|$, since $S_i \subset R(S, i)$.  Furthermore, we
        have by the inductive hypothesis that $|\Delta R(S, \SEst, i, i)| = o(\mu_{i-1}(S)^{0.66}) = o(|S_i|^{0.66})$.  Thus, we have
        \begin{align}
            p_e |\Delta R(S, \SEst, i, i)| = o(p_e |R(S, i)|),
            \label{expr:RDiscTermFinalUB}
        \end{align}
        with (conditional) probability $1$, as desired.
        
        \paragraph{\textbf{Upper bounding (\ref{expr:QDiscTerm})} by $\sum_{v \in \Frontier_i(S)} Q(i, S, v) p \mu_{i-1}^{2/3}(S)$:}
        
        To upper bound (\ref{expr:QDiscTerm}), we apply the triangle inequality and extend both
        sums to $v$ in $\Frontier_i(S) \lunion \Frontier_{i}(\SEst)$.  This results in the following upper bound: 
        \begin{align}
            (\ref{expr:QDiscTerm})
            \leq \sum_{v \in \Frontier_i(S) \lunion \Frontier_{i}(\SEst)} | Q(i, S, v) - Q(i, \SEst, v)|.
            \label{expr:QDiscUpperBound1}
        \end{align}
        To proceed, we will upper bound the number of nonzero terms in (\ref{expr:QDiscUpperBound1}).  Each nonzero term can be upper bounded by $1$, since
        $Q(i, S, v), Q(i, \SEst, v)$ are both probabilities.  We will show that the number
        of nonzero terms is at most $O(|\Frontier_i(S)|p\cdot \mu_{i-1}^{2/3}(S))$ with high probability.
        
        We write
        \begin{align}
            &Q(i, S, v) - Q(i, \SEst, v) \\
            &= p_e + (1-p_e)(1 - (1-p_n)^{\deg_{S_i}(v)})
            - p_e - (1 - p_e)(1 - (1-p_n)^{\deg_{\SEst_i}(v)}) \\
            &= (1 - p_e)( (1-p_n)^{\deg_{\SEst_i}(v)} -  (1-p_n)^{\deg_{S_i}(v)}).
        \end{align}

        Thus, a term in the sum (\ref{expr:QDiscUpperBound1}) is nonzero if and only
        if $\deg_{S_i}(v) \neq \deg_{\SEst_{i}}(v)$.  This happens if and only
        if $v$ has at least one edge to some vertex in $\SEst_{i} - S_i$.
        Thus, our task reduces to figuring out how many vertices $v$ there are that connect
        to some element of $\hat{C}_{i} - C_i$.  The expected number of such vertices
        is 
        \begin{align}
            |\Frontier_i(S) \lunion \Frontier_{i}(\SEst)| \cdot (1-q^{|\SEst_{i} - S_i|}).
            \label{expr:QDifferenceUpperBound}
        \end{align}
        This is an upper bound on the contribution of (\ref{expr:QDiscTerm}).
        We thus have 
        \begin{align}
            (\ref{expr:QDiscTerm})
            \leq |\Frontier_i(S) \lunion \Frontier_{i}(\SEst)| \cdot (1-q^{|\SEst_{i} - S_i|}).
        \end{align}

        Next, we show that $|\Frontier_i(S) \lunion \Frontier_{i}(\SEst)| = O(|\Frontier_i(S)|)$.
        To do this, we apply the results from upper bounding (\ref{expr:FrontierDiscTerm}).
        In particular, 
        \begin{align}
            |\Frontier_i(S) \lunion \Frontier_{i}(\SEst)|
            = |\Frontier_i(S) \lintersect  \Frontier_{i}(\SEst)| + |\Delta \Frontier(S, \SEst, i, i)|
            \leq |\Frontier_i(S)| + |\Delta \Frontier(S, \SEst, i, i)|
            = O(|\Frontier_i(S)|).
        \end{align}
        Next, we show that $1 - q^{ |\SEst_{i} - S_i| } = p\mu_{i-1}^{2/3}(S)$.  
        We can write
        \begin{align}
            q^{ |\SEst_{i} - S_i| }
            = (1 - p)^{|\SEst_{i} - S_i| }
            \sim
            e^{-p|\SEst_{i} - S_i|},
        \end{align}
        provided that $p\cdot |\SEst_{i} - S_i| = o(1)$.  Now
        from the inductive hypothesis, $|\SEst_{i} - S_i| = O(|S_i|^{2/3})$, and from Lemma~\ref{lemma:mu_i-growth}, we know
        that $|S_i| = O((np)^{i})$. 
        Then we have that
        \begin{align}
            1 - q^{ |\SEst_{i} - S_i| }
            \leq %
            1 - e^{-O( p(np)^{2/3 i} )}.
        \end{align}
        In order for this second term to be $1 - o(1)$, it is sufficient to have that
        \begin{align}
            p^{i+1} n^{i} = o(1).
        \end{align}
        This happens if and only if
        \begin{align}
            p^{i+1} = o(n^{-i})
            \iff p = o(n^{-\frac{i}{i+1}}).
        \end{align}
        This is guaranteed by our assumption that $p = o(n^{-\frac{T}{T+1}})$.
        Thus,
        \begin{align}
            1 - q^{ |\SEst_{i} - S_i| }
            \leq %
            1 - e^{-O( p(np)^{2/3 i} )}
            \sim p \mu_{i-1}^{2/3}(S).
        \end{align}
        
        We have shown that
        \begin{align}
            (\ref{expr:QDiscTerm})
            = O(|\Frontier_i(S)| \dot p \mu_{i-1}^{2/3}(S)).
        \end{align}
        Next, we show that $\sum_{v \in \Frontier_i(S)} Q(i, S, v) = \Omega(|\Frontier_i(S)|)$.
        We have
        \begin{align}
            Q(i, s, v) %
            \geq 1 - (1 - p_n)^{\deg_{S_i}(v)}.
        \end{align}
        Since the sum is over $v \in \Frontier_i(S)$, this implies that
        $\deg_{S_i}(v) \geq 1$.  So
        \begin{align}
            Q(i, C, v)
            \geq 1 - (1 - p_n)
            = p_n
            = \Omega(1).
        \end{align}
        Thus, 
        \begin{align}
            \sum_{v \in \Frontier_i(S)} Q(i, C, v)
            \geq |\Frontier_i(S)| \cdot p_n
            = \Omega(|\Frontier_i(S)|).
        \end{align}

        Thus, we have shown that
        \begin{align}
            (\ref{expr:QDiscTerm})
            \leq const \cdot \sum_{v \in \F_i(C)} Q(i, C, v) \cdot (1 - q^{|\hat{C}_{i} - C_i|})
            = const \sum_{v \in \F_i(C)} Q(i, C, v) p \mu_{i-1}^{2/3}(C), %
            \label{expr:QDiscTermFinalUB}
        \end{align}
        with conditional probability at least $1 - e^{-\Omega((np)^i)}$.

        \paragraph{Completing the proof}
        We combine (\ref{expr:QDiscTermFinalUB}), (\ref{expr:FDiscTermFinalUB}),
        and (\ref{expr:RDiscTermFinalUB}) to complete the proof.  %
    \end{proof} %

    So we have that the difference between $\mu_i(C)$ and 
    $\mu_{i}(\hat{C})$ is negligible in relation to $\mu_{i}(C)$.

Now, the next two propositions give the base case and inductive
step of the proof of Theorem~\ref{thm:fastclock-utility}.

    \begin{proposition}[Base case of the proof of Theorem~\ref{thm:fastclock-utility}]
        We have that, with probability $1$,
        \begin{align}
            |\Delta R_0| = |\Delta R(S, \SEst, 0, 0)| = 0,
        \end{align}
        and
        \begin{align}
            |\SEst_{0} \symmdiff S_{0}|
            = 0.
        \end{align}
        \label{prop:base-case}
    \end{proposition}
    \begin{proof}
        This follows directly from the assumed initial 
        conditions.
    \end{proof}

    \begin{proposition}[Inductive step of the proof of Theorem~\ref{thm:fastclock-utility}]
        Assume that the inductive hypotheses (\ref{expr:R_inductive}) and (\ref{expr:C_inductive}) hold for $i$.  Then we have the following:
        \begin{align}
            |\SEst_{i+1} \symmdiff S_{i+1}|
            \leq |\Delta R_i| = |\Delta R(S, \SEst, i, i)| = o(\mu_{i-1}(S)^{2/3}) = o(\mu_i(S)^{2/3}). %
        \end{align}
        Equivalently,
        \begin{align}
            1 - \frac{|\SEst_{i+1} \lintersect S_{i+1}|}{|S_i|}
            = o(\mu_{i}(S)^{2/3}).
        \end{align}
        Furthermore,
        \begin{align}
            |\Delta R_{i+1}|
            \leq |\Delta R_{i}| \leq o(\mu_{i-1}(S)^{2/3}) = o(\mu_{i}(S)^{2/3}).
        \end{align}
        In other words, both inductive hypotheses are satisfied for
        $i+1$.  This holds with probability at least
        $1 - e^{-\Omega(\mu_{i}(C))}$.
        \label{prop:inductive-step}
    \end{proposition}
    \begin{proof}
        To prove this, we first need a few essential inequalities. 
        
        \begin{itemize}
        \item
        By definition of the algorithm,
        \begin{align}
            |\SEst_{i+1}| 
            \leq \mu_{i}(\SEst) (1 + \mu_{i}(\SEst)^{-1/3}),
            \label{expr:hatCmuUpperBound}
        \end{align}
        with probability $1$.  
        
        \item
        We will also need to prove an upper bound on 
        $|\SEst_{i+1}| - |S_{i+1}|$.  In particular, we
        will show that with probability at least $1 - e^{-\Omega(\mu_i(S))}$,
        \begin{align}
            |\SEst_{i+1}|
            \leq |S_{i+1}| \cdot (1 + O(\mu_i(S)^{-1/3})).
            \label{expr:C_hatCUpperBound}
        \end{align}
        We show this as follows.
        From 
        Theorem~\ref{thm:mu_diff_upper_bound},
        \begin{align*}
            \mu_{i}(\SEst)
            \leq \mu_i(S)(1 + p \mu_{i-1}(S)^{2/3}),
        \end{align*}
        with probability $\geq 1 - e^{-\Omega(\mu_i(S))}$.
        This implies, via (\ref{expr:hatCmuUpperBound}), 
        that
        \begin{align*}
            |\SEst_{i+1}|
            \leq \mu_i(S) \cdot (1 + p\mu_{i-1}(S)^{2/3})
            \cdot (1 + O(\mu_i(S)^{-1/3})).
        \end{align*}
        By concentration of $|S_{i+1}|$, with probability
        at least $1 - e^{-\Omega(\mu_i(S))}$,
        this is upper bounded as follows:
        \begin{align*}
            |\SEst_{i+1}|
            \leq |S_{i+1}| (1 + O(|S_{i+1}|^{-1/2+const})) (1 + p\mu_{i-1}(S)^{2/3})(1 + O(\mu_{i}(S)^{-1/3})).
        \end{align*}
        Now, we can see from (\ref{expr:pnPowerAsymptotics})
        that this is equal to the desired upper bound.
        We have thus shown (\ref{expr:C_hatCUpperBound}).
        \end{itemize}

        Now, with the preliminary inequalities proven, we
        proceed to prove the proposition.
        We split into two cases:
        
        \begin{itemize}
            \item
                \textbf{$S_{i+1}$ begins before $\SEst_{i+1}$
                (in other words, $|R(S, i)| < |R(\SEst, i)|$).}
                
                In this case, we will show (i) that $S_{i+1}$ must end before $\SEst_{i+1}$ (i.e., that $|R(S, i+1)| \leq |R(\SEst, i+1)|$) with high
                probability, 
                (ii) that
                \begin{align}
                    \Delta R_{i+1}
                    = 0, %
                \end{align}
                and 
                (iii) that
                \begin{align}
                    |\SEst_{i+1} \symmdiff S_{i+1}|
                    \leq |\Delta R_i|. %
                \end{align}

                To show that (i) is true, we note that because $S_{i+1}$ begins before
                $\SEst_{i+1}$, 
                $S_{i+1}$ consists of an initial segment $\hat{S}_{j_1}, \hat{S}_{j_1+1}, ..., \hat{S}_{j_2}$  with total cardinality 
                $|\Delta R_i|$, ending in an observation endpoint (specifically, the one corresponding to $R(\SEst, i)$), 
                followed by a segment $\hat{S}_{j_2+1}, ..., \hat{S}_{j_3}$ of total c
                ardinality $|S_{i+1}| - |\Delta R_i|$,
                again ending in an observation endpoint.  This is true
                by definition of $\Delta R_i$.  The second segment begins
                at the same point as $\SEst_{i+1}$ (that is, $R(\SEst, i) = \union_{j=0}^{j_2+1} \hat{S}_j$), and we know
                that it has cardinality
                \begin{align}
                    |S_{i+1}| - |\Delta R_i|
                    \leq |S_{i+1}|
                    \leq \mu_{i}(S)(1 + \mu_i(S)^{-1/2+const})
                    \leq \mu_{i}(\SEst)(1 + \mu_{i}(\SEst)^{-1/3}),
                \end{align}
                by concentration of $|S_{i+1}|$.  The last inequality follows from the fact that $\mu_i(S)
                = \Theta(\mu_i(\SEst))$.
                Thus, the second segment of $S_{i+1}$ must be contained
                in $\SEst_{i+1}$, by (\ref{expr:hatCmuUpperBound}), by definition of the FastClock algorithm, as desired.
                
                This has the following implication: we can express
                $|\Delta R_{i+1}|$ as 
                \begin{align}
                    |\Delta R_{i+1}|
                    = |\SEst_{i+1}| - (|S_{i+1}| - |\Delta R_{i}|)
                    \leq \mu_i(S)^{-1/3}  + |\Delta R_{i}|.
                \end{align}
                We have used (\ref{expr:C_hatCUpperBound}).
                Since, by the inductive
                hypothesis, we have $|\Delta R_i| = o(\mu_{i-1}(C)^{0.66})$,
                and since this is $o(|C_{i+1}|)$, we have that 
                \begin{align}
                    |\Delta R_{i+1}|
                    = 0,
                \end{align}
                by Assumption~\ref{asmpt:observation-rate} that no observation interval has too
                few vertices.  This follows because, if $\Delta R_{i+1}$ were nonempty, then
                it would contain an observation interval (i.e., $\hat{S}_j$ for some $j$) with cardinality at most $o(|S_{i+1}|)$ that is a subset of $S_{i+1}$.  This contradicts Assumption~\ref{asmpt:observation-rate}.  Thus, we have established (ii).

                We next show (iii).  We have
                \begin{align}
                    |\SEst_{i+1} \symmdiff S_{i+1}|
                    \leq |\Delta R_i|,
                \end{align}
                by the fact that  
                $|\SEst_{i+1} \symmdiff S_{i+1}|
                = |\Delta R_i| + |\Delta R_{i+1}|$.

            \item
                \textbf{Or $S_{i+1}$ begins after or at the same time as $\SEst_{i+1}$
                (in other words, $|R(S, i)| \geq |R(\SEst, i)|$).}

                In this case, we will show (i) that
                \begin{align}
                    |\Delta R_{i+1}| = 0,
                \end{align}
                and (ii) that
                \begin{align}
                    |\SEst_{i+1} \symmdiff S_{i+1}|
                    \leq |\Delta R_{i}|.
                \end{align}
                
                This is because of the following identity:
                \begin{align}
                    |\SEst_{i+1}|
                    = |\Delta R_i| + |S_{i+1}| + |\Delta R_{i+1}| I_{i+1},
                    \label{expr:special-identity}
                \end{align}
                where 
                \begin{align}
                    I_{i+1}
                    = \begin{cases}
                        1 & \text{ $S_{i+1}$ stops before $\SEst_{i+1}$} \\
                        -1 & \text{otherwise}
                    \end{cases}
                \end{align}
                This is a consequence of the following derivation,
                which relies on the definitions of all involved
                terms.
                \begin{align*}
                    |\Delta R_i| + |S_{i+1}| + |\Delta R_{i+1}| I_{i+1}
                    &= \sum_{k=0}^{i} |S_k| - \sum_{k=0}^i |\SEst_k|
                    + |S_{i+1}| + \left|\sum_{k=0}^{i+1} |S_k| - \sum_{k=0}^{i+1} |\SEst_k|\right| I_{i+1} \\
                    &= \sum_{k=0}^{i+1} |S_{k}| - \sum_{k=0}^i |\SEst_k|  - \left(\sum_{k=0}^{i+1} |S_k| - \sum_{k=0}^{i+1} |\SEst_k|\right) \\
                    &= |\SEst_{i+1}|.
                \end{align*}

                Rearranging (\ref{expr:special-identity}) to solve for $|\Delta R_{i+1}|$, we have that
                \begin{align*}
                    |\Delta R_{i+1}|
                    &= | |\SEst_{i+1}| - |\Delta R_{i}| - |S_{i+1}| | \\
                    &\leq | |\SEst_{i+1}| - |S_{i+1}| | + |\Delta R_{i}| \\
                    &= | |\SEst_{i+1}| - |S_{i+1}| | + o(\mu_{i-1}(S)^{2/3}) \\
                    &\leq O(\mu_{i}(S)^{-1/3}) + o(\mu_{i-1}(S)^{2/3}) \\
                    &= o(\mu_i(S)).
                \end{align*}
                Here, we have used the triangle inequality and the
                inductive hypothesis on $|\Delta R_{i}|$, followed
                by the inequality (\ref{expr:C_hatCUpperBound}).
                
                Since $|\Delta R_{i+1}| = o(\mu_{i}(C))$, it must be $0$ because of Assumption~\ref{asmpt:observation-rate}, which
                verifies the inductive hypothesis on $|\Delta R_{i+1}|$.  
                
                Furthermore, this implies that
                \begin{align}
                    |\SEst_{i+1} \symmdiff S_{i+1}|
                    \leq |\Delta R_{i}|,
                \end{align}
                which verifies the inductive hypothesis on $|\SEst_{i+1} \symmdiff S_{i+1}|$.
        \end{itemize}
        
        The inductive hypotheses follow directly from
        the above.
    \end{proof}

We can now prove the utility theorem, Theorem~\ref{thm:fastclock-utility}.
\begin{proof}[Proof of Theorem~\ref{thm:fastclock-utility}]
    Let $B_i$ denote the \emph{bad} event that either inductive hypothesis
    fails to hold at step $i$.  We will lower bound
    \begin{align}
        \Pr[ \intersect_{i=0}^{T-1} \neg B_i].
    \end{align}
    By the chain rule, we have
    \begin{align}
        \Pr[ \intersect_{i=0}^{T-1} \neg B_i]
        = \Pr[\neg B_0] \prod_{i=1}^{T-1} \Pr[ \neg B_i ~|~ \intersect_{j=0}^{i-1} \neg B_j].
    \end{align}
    
    From Proposition~\ref{prop:inductive-step}, 
    Proposition~\ref{prop:base-case}, and Lemma~\ref{lemma:mu_i-growth}, 
    this is lower bounded by
    \begin{align}
        \prod_{i=1}^{T-1} (1 - e^{-D \cdot (np)^{i+1}})
        &= \exp\left( \sum_{i=1}^{T-1}\log\left( 1 - e^{-D(np)^{i+1}} \right)\right) \\
        &= \exp\left(  -\sum_{i=1}^{T-1}e^{-D(np)^{i+1}}\cdot (1 + o(1))  \right) \\
        &= 1 - e^{-\Omega(np)}.
    \end{align}
    
    Now, the event that none of the bad events hold implies the claim,
    which completes the proof.
\end{proof}

With Theorem~\ref{thm:fastclock-utility} in hand, we can
prove the main result, Theorem~\ref{thm:fastclock-approx-main}.
\begin{proof}[Proof of Theorem~\ref{thm:fastclock-approx-main}]
    Let us recall the definition of $d_{\hat{S}}(C, \hat{C})$.  We have
    \begin{align}
        d_{\hat{S}}(C, \hat{C})
        = \frac{1}{{n\choose 2}} \sum_{i < j} \Dis_{C,\hat{C}}(i,j).
    \end{align}
    What we need is an upper bound on this quantity in terms of the error term
    $f(n) = (np)^{-1/3}$ in Theorem~\ref{thm:fastclock-utility}.
    To this end, we partition the
    sum according to vertex membership in clock intervals as follows:
    \begin{align}
        {n\choose 2} d_{\hat{S}}(C, \hat{C})
        =
        \sum_{k_1=1}^{|S|} \sum_{i<j \in S_{k_1}} \Dis_{C,\hat{C}}(i,j) 
        + \sum_{k_1=1}^{|S|} \sum_{k_2=k_1+1}^{|S|} \sum_{i \in S_{k_1}, j \in S_{k_2}} \Dis_{C,\hat{C}}(i,j).
    \end{align}
    In the first sum, $i$ and $j$ are \emph{not} ordered by $C$, because they lie in the
    same set in $S$.  We consider the corresponding set in $\SEst$.
    From the theorem, at least ${|C_{k_1}|\cdot (1-f(n))\choose 2}$ vertex pairs from
    $S_{k_1}$ are correctly placed together in $\SEst$.  
    Furthermore, at least
    \begin{align}
        |S_{k_1}| \cdot (1 - f(n)) \cdot \sum_{k_2=k_1+1}^{|S|} (1-f(n))|S_{k_2}| %
    \end{align}
    pairs of vertices with one vertex in $S_{k_1}$ are correctly placed in different intervals.
    So the number of correctly ordered/unordered vertex pairs is at least
    \begin{align}
        \sum_{k_1=1}^{|S|} \left( \frac{|S_{k_1}|^2 \cdot (1 - f(n))^2}{2} + \sum_{k_2=k_1+1}^{|S|} |S_{k_1}||S_{k_2}| (1 - f(n))^2 \right)
        \sim {n\choose 2} \cdot (1 - f(n))^2.
    \end{align}
    Since $f(n) = o(1)$, this is asymptotically equal to 
    ${n\choose 2}\cdot (1 - 2f(n))$.
    
    This completes the proof.
\end{proof}

\subsection{Proof of Theorem~\ref{thm:fastclock-running-time}}
We analyze the worst-case running time of FastClock as follows:
initialization takes $O(1)$ time.  The dominant contribution to the running time is
the \emph{while} loop.  Since $t_{obs}$ is initially $0$ and increases by at least
$1$ in each iteration, the total number of iterations is at most $N$.  The remaining
analysis involves showing that each vertex and edge is only processed, a constant number of 
times, in $O(1)$ of these loop iterations, so that the running time is at most 
$O(N + n + m)$, as claimed.

In particular, the calculation of $\mu_t$ in every step involves a summation over all
edges from currently active vertices to their uninfected neighbors, along with a calculation involving the current number of uninfected vertices (which we can keep track of using $O(1)$
calculations per iteration of the loop).  A vertex is only active
in a single iteration of the loop.  Thus, each of these edges is only processed once in
this step.
The calculation of $t'_{obs}$ entails calculating a sum over elements of $\hat{S}$
that are only processed once in all of the iterations of the loop.  The calculation of
all of the $|\hat{S}_i|$ can be done as a preprocessing step via an iteration over all 
$n$ vertices of $G$.  Finally, the calculation of $\Frontier_{t+1}$ entails a union
over the same set of elements of $\hat{S}$ as in the calculation of the maximum, followed
by a traversal of all edges incident on elements of $\SEst_{t+1}$ whose other ends connect
to uninfected vertices.  These operations involve
processing the vertices in $\SEst_{t+1}$ (which happens only in a single iteration of
the loop, and, thus, with the preprocessing step of calculating the $|\hat{S}|_i|$, only
a constant number of times in the entire algorithm).  The edges leading to elements of
$\Frontier_{t+1}$ from elements of $\SEst_{t+1}$ are traversed at most twice in the loop:
once in the building of $\Frontier_{t+1}$ and once in the next iteration in the calculation
of $\mu_t$.

This implies that each vertex and edge is only processed $O(1)$ times in the entire algorithm.
This leads to the claimed running time of $O(N+n+m)$, which completes the proof.

\end{document}